\documentclass[12pt]{article}
\usepackage[utf8]{inputenc}

\usepackage{booktabs}
\usepackage{threeparttable}
\usepackage{color}
\usepackage{float}
\usepackage{amsmath}
\usepackage{amssymb}
\usepackage{graphicx}
\usepackage{booktabs}
\usepackage{geometry}
\usepackage[authoryear]{natbib}
\usepackage[onehalfspacing]{setspace}

\newtheorem{definition}{Definition}\newtheorem{lemma}{Lemma}\newtheorem{proposition}{Proposition}\newenvironment{proof}[1][Proof]{\noindent\textbf{#1.} }{\ \rule{0.5em}{0.5em}}

\usepackage{xcolor}
\usepackage{subcaption}
\usepackage{enumitem}
\usepackage[colorlinks=true]{hyperref}
\hypersetup{
 linkcolor=blue,citecolor=black}
 
\title{From Summer to Spring: A Shift in US Housing Market Seasonality\thanks{Hu: \texttt{yh623@cam.ac.uk}, Selcuk: \texttt{selcukc@cardiff.ac.uk}}}
\author{
  Yihan Hu \\ \small University of Cambridge 
  \and
  Cemil Selcuk \\ \small Cardiff University
}
\date{}

\begin{document}
\maketitle
\begin{abstract}
The US housing market exhibits pronounced seasonal cycles: prices and sales rise
through spring, peak in summer, and decline through autumn and winter. Since 2021,
this pattern has shifted earlier in the calendar year, with spring strengthening at the
expense of the traditional summer peak. A leading explanation for housing market
seasonality is the search-and-matching model of Ngai and Tenreyro (2014), which links
these cycles to household mobility through a thick-market mechanism. In this framework,
periods with higher mobility generate thicker markets and higher prices and transaction
volumes. Viewed through this lens, a shift in the seasonal cycle of prices and sales raises the question of whether the timing of household moves has changed.
Did residential mobility shift earlier in the calendar year after 2021?

We find that it did. Using SIPP data, and corroborating evidence from Google Trends
indicators, we document a post-2021 shift in mobility toward spring. We extend the model
to a monthly frequency, prove the existence and uniqueness of the equilibrium, and calibrate
it to the observed mobility patterns. The calibrated model reproduces the spring shift in
both prices and transaction volumes, consistent with the view that a change in the timing
of household mobility alone can account for the recent shift in housing market seasonality.
\end{abstract}
 
\section{Introduction}

The US housing market exhibits highly regular seasonal cycles. Prices and transaction volumes rise through spring, peak in summer, and fall through autumn and winter. This pattern has been documented repeatedly over several decades and across a wide range of 
markets \citep{Case_Shiller1989, Goodman1993}. The
regularity of this cycle has made it one of the most robust
stylised facts in housing economics. Summer is the ``hot'' season,
winter the ``cold'' one, and the amplitudes of the swings are economically meaningful in both prices and volumes.

There are signs, however, that this regularity may have changed. Since the
COVID-19 pandemic, seasonal peaks appear to arrive earlier in the
calendar year, with spring months strengthening at the expense of
the traditional summer surge, and the amplitude of the cycle
itself may have changed \citep{Malpezzi2023, HuHuang2025}. These changes represent a break in one of the most stable
empirical patterns in housing economics.

As a first step, we verify these shifts: using transaction data from Zillow, one of the largest real estate
platforms in the United States, we establish that the seasonal
cycle has indeed shifted, in both prices and sales. In the
post-2021 period, the spring months pick up earlier and more
strongly than they did before, while the summer term softens. The
seasonal cycle has not disappeared; it has moved forward in the
calendar year. The spring quarter gains roughly 2 percentage points in seasonal
price intensity and 6 percentage points in sales, with
similar and statistically significant contractions in summer, confirming
the earlier findings.

To understand what might be driving this shift, we turn to theory. In \citet{ngai_hot_2014}, arguably the leading theoretical account of housing market seasonality, the seasonal cycle in housing markets is driven by seasonal variation in household mobility. More households tend to move in summer than in winter. When many households move at the same time, more buyers and sellers enter the housing market simultaneously, making the market thicker. In a thick market, buyers are more likely to find a house that matches their preferences, which raises the probability of a transaction. Because matches are better on average, the surplus from trade is larger, and prices are higher as well. In winter, when fewer households move, the market becomes thin: matches are harder to find, transactions are less frequent, and prices are lower. In this way, the seasonal patterns of prices and sales are pinned down directly by the seasonal pattern of household mobility.

This mechanism has a sharp implication. The pattern of prices and transactions is inherited from the pattern of household mobility. If prices and sales have shifted earlier in the calendar year, then, all else equal, the model predicts that households themselves must be moving earlier as well. Did the timing of household moves shift earlier after 2021?

We find that it did. Our primary evidence comes from the Survey of
Income and Program Participation (SIPP), a nationally representative
longitudinal survey conducted by the US Census Bureau. SIPP follows
households over time and records detailed information on employment,
income, and residential mobility. Crucially for our purposes, respondents
report the calendar month in which they moved residence, allowing the
timing of household moves to be measured directly. The data reveal a
clear shift. June, historically the busiest moving month, loses 3.0
percentage points of annual moves post-2021, and July drops a further
1.9 percent. On the other side, March gains 1.8 percent and April gains 1.4 percent. In aggregate, the spring
quarter absorbs 4.0 additional percentage points of annual moves
post-2021, while the summer quarter contracts by 4.8 percent.
All changes are statistically significant.

We corroborate this finding with Google Trends search indices for eight keywords that capture different aspects of a household move: rental
trucks (\textit{moving truck}, \textit{uhaul}), professional movers
(\textit{atlas van lines}, \textit{mayflower moving}), packing
(\textit{moving supplies}), moving assistance (\textit{moving help}),
and administrative tasks (\textit{address forwarding},
\textit{thumbtack}). All eight show a statistically significant shift
of seasonal search activity from summer toward spring. Finally, a
recent report by MoveHQ, a leading logistics company that handles
household relocations across the US, documents the same pattern in its
own operational records covering over one million moves. Three
independent sources of evidence, each measuring different populations
using different methods, point to the same pattern: American households
are moving earlier in the year.

To link mobility patterns to prices and sales, we extend the setup in
NT from two seasons to twelve calendar months. The monthly frequency
is essential because the monthly shifts we document cannot be captured in a model that operates
at a semi-annual frequency. We prove the existence and uniqueness of
the housing market equilibrium and propose an algorithm based on
contraction mapping to numerically solve it. 

We then feed the pre- and post-2021 hazard rates obtained from the SIPP mobility profiles into the model calibration. Three features of the results are worth highlighting. First, the hump-shaped seasonal pattern of mobility translates into a hump-shaped profile of prices and transaction volumes, matching the qualitative pattern observed in the data. Market activity and prices are highest in summer, lowest in winter, and intermediate in spring and autumn.

Second, and more importantly, the observed shift in seasonal mobility produces a corresponding shift in the trajectory of prices and sales. As households begin to move earlier in the year, the housing market thickens earlier, raising matching rates and prices in spring at the expense of summer. In other words, the observed spring shift in both prices and transaction volumes is quantitatively consistent with a change in the timing of household mobility alone, requiring no change in preferences, technology, or market structure.

Third, the calibrated numbers broadly align with the price and transaction data. In the Zillow data, seasonal price deviations are modest, ranging from roughly $-7$ to $+6$ percent around the annual mean, while sales volumes swing far more dramatically, from around $-30$ to $+21$ percent. The model calibrated on SIPP moving hazard rates reproduces the price variation well, generating a seasonal range of $-5$ to $+8$ percent, in the right ballpark. It overshoots on volumes, however, producing swings of $-37$ to $+51$ percent, roughly twice the amplitude seen in the data. When we instead calibrate using Google Trends data on searches for \textit{moving truck}, the volume range narrows to $-41$ to $+35$ percent. This still remains more volatile than the Zillow benchmark, but it is meaningfully closer.

Overall, the calibrated model reproduces the key qualitative patterns in the data. When the seasonal distribution of household mobility shifts toward spring, the model generates an earlier rise in both housing prices and transaction activity, moving the seasonal cycle forward in the calendar year. Quantitatively, although it tends to overstate the magnitude of volume fluctuations, the calibration produces reasonable price amplitudes and the correct direction of seasonal shifts.

Why did the timing of household moves shift in the first place? This is a substantive question in its own right, and we do not attempt to resolve it in this paper. A natural candidate is the pandemic era surge in remote work, which stabilised at roughly 25-30$\%$ per cent of paid workdays from 2021 onwards, up from 5-7$\%$ before the pandemic \citep{BarreroBloomDavis2023}. By weakening the link between residence and workplace, work-from-home arrangements may have given households greater discretion over the timing of their relocations and loosened the constraints that historically anchored their moves to summer. But this remains a conjecture. The existing literature on post-pandemic mobility is almost entirely concerned with where households moved, such as the ``donut effect'', and the suburbanisation of demand. The question of why the timing changed remains open. We hope the documentation and modelling in this paper make that a more tractable question for future work.
\subsection{Related literature}\label{sec:litreview}
The observation that house prices and transaction volumes display pronounced within-year cycles has a long empirical history. \citet{Case_Shiller1989} note seasonal patterns in repeat-sale price indices. \citet{Goodman1993} documents that residential moves concentrate in summer and identifies the school calendar as a key driver, modelling the concentration of moves as the equilibrium outcome of a housing-market matching process; however, he notes that families with school-age children account for less than a third of all movers, too few to explain the full amplitude of the cycle on their own. More recently, \citet{MillerEtAl2013} estimate month-of-year pricing patterns using hedonic methods across a large cross-section of US metropolitan areas and find significant price variations during the year for most months and most markets.
 
Arguably, the most influential theoretical explanation of this phenomenon is due to \citet{ngai_hot_2014}, who show that these predictable cycles can arise from exogenous seasonal variation in household mobility.\footnote{Several papers, mostly building on search-and-matching theory, study housing market dynamics from a theoretical perspective; see, for instance, \citet{Krainer2001}, \citet{NovyMarx2009}, \citet{PiazzesiSchneider2009}, \citet{DiazJerez2013}, and \citet{Selcuk2014}. These papers focus on the response of prices and sales to unexpected shocks. In contrast, the framework in NT explains predictable seasonal cycles in prices and transactions as the equilibrium outcome of predictable cycles in household mobility.} In their search-and-matching framework, seasonal increases in mobility make the market thicker by bringing more buyers and sellers into the market at the same time. Thicker markets improve matching opportunities and raise the gains from trade, generating higher prices and transaction volumes during the hot season. \citet{GenesoveHan2012} confirm empirically that housing market liquidity is indeed responsive to market thickness, with sellers transacting faster and buyers searching less when more participants are active. Our paper recasts the NT framework from two seasons to twelve months (the original semi-annual version cannot capture a shift from summer to spring) and uses it to study the post-pandemic shift in mobility. We feed the observed change in monthly moving patterns into the model and examine the implied effects on prices and transactions.
 
A large and growing literature examines how the pandemic reshaped housing markets, with particular attention to the role of remote work in altering the spatial distribution of demand. \citet{GuptaEtAl2022} show that the pandemic flattened the bid-rent curve in US metropolitan areas, with prices and rents rising faster in suburban locations where remote work is more prevalent. \citet{RamaniBloom2021} document the ``donut effect,'' in which population and economic activity shifted from city centres to surrounding areas. \citet{MondragonWieland2022} provide evidence that remote work was a key driver of the post-2020 surge in aggregate housing demand. In contrast, our paper asks not where or how much households moved, but when, and how the timing of those moves shaped seasonal patterns in prices and sales. \citet{HuHuang2025} provide the closest precedent on this dimension, documenting post-pandemic changes in the timing and intensity of seasonal peaks. Our contribution is to embed such a shift in a structural search-and-matching framework and trace it to a similar seasonal shift in household mobility.
 
The rest of the paper proceeds as follows. Section~\ref{sec: document shifts} documents the shift in housing market seasonality using Zillow data. Section~\ref{sec:moves} establishes the corresponding shift in household moving patterns based on SIPP, MoveHQ, and Google Trends. Section~\ref{sec:model} presents the monthly search-and-matching model. Section~\ref{sec:calibration} calibrates the model and  Section~\ref{sec:conclusion} concludes.

\section{Seasonality in Prices and Sales} \label{sec: document shifts}

In this section we document the seasonality in house prices and sales
volumes, and how it changed after 2021. The data come from Zillow, a
large real estate platform that aggregates housing market information
derived from public records across the United States. We use two series:
the median sale price and the count of completed sales, both at monthly
frequency for the United States as a whole. The sample runs from February
2008 to June 2025, giving 209 monthly observations. We deflate prices by
the all-items CPI (not seasonally adjusted, FRED series CPIAUCNS) and
normalise so that the 2019 average equals 100.

\begin{figure}[htbp]
    \centering
    \begin{subfigure}{0.45\textwidth}
        \includegraphics[width=\linewidth]{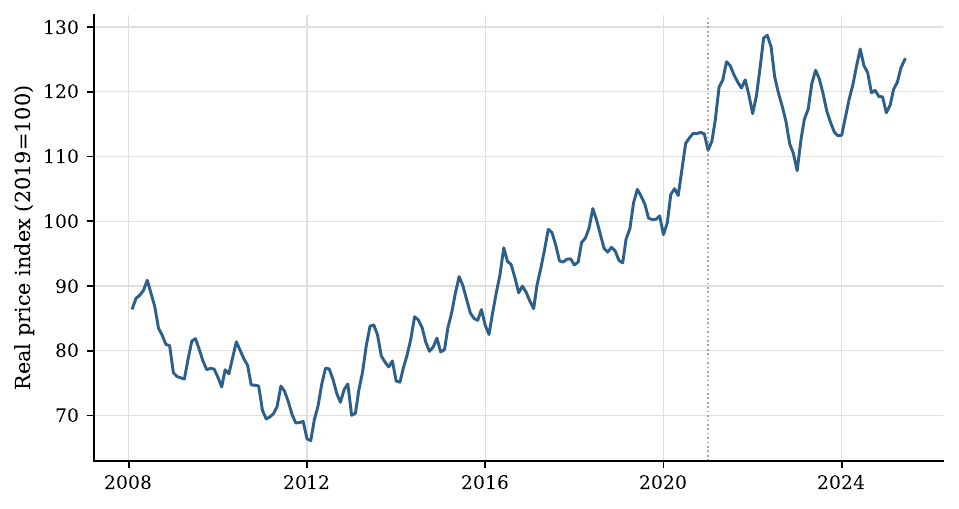}
        \caption{Real prices}
    \end{subfigure}\hfill
    \begin{subfigure}{0.45\textwidth}
        \includegraphics[width=\linewidth]{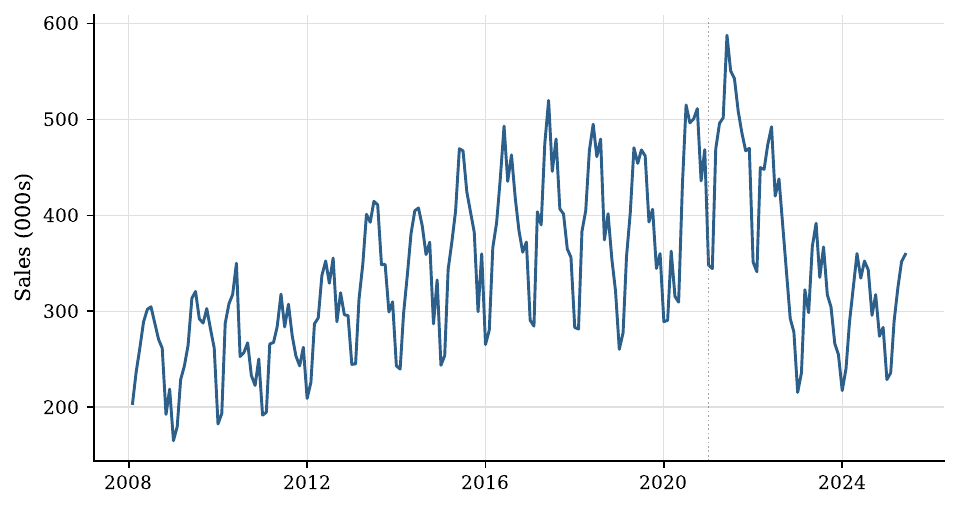}
        \caption{Sales}
    \end{subfigure}
    \caption{\small US housing market: real prices and sales (Zillow)}
    \label{fig:priceandsalesraw}
\end{figure}

Figure~\ref{fig:priceandsalesraw} plots the raw series for real house
prices and sales. Both series display clear and regular within-year
fluctuations, with prices and transactions rising during part of the year
and falling during another. To isolate the seasonal component, let $y_t$
denote the outcome observed at time $t$, where $T$ and $m \in
\{1,\dots,12\}$ index the calendar year and month of observation $t$
respectively. The standard approach is to extract
seasonal components via the two-way fixed-effects specification \citep{lovell_1963}: 
\begin{equation}
y_t = \alpha_T + \gamma_m + \epsilon_t,
\label{eq:fe_seasonal}
\end{equation}
where $\alpha_T$ are year fixed effects and $\gamma_m$ are month-of-year
fixed effects. The year effects absorb all annual movements in levels, so
the month effects $\gamma_m$ are identified entirely from within-year
variation. Imposing $\sum_{m=1}^{12}\gamma_m = 0$, the coefficient
$\gamma_m$ measures the average deviation of month $m$ from the annual
mean. In practice, we compute the seasonal component as
\begin{equation}
d_{m,T} = 100\,\frac{y_{m,T} - \bar{y}_T}{\bar{y}_T},
\label{eq:dmt}
\end{equation}
the percentage deviation of month $m$ in year $T$ from that year's annual
mean $\bar{y}_T$. By the Frisch--Waugh--Lovell theorem, subtracting each
year's mean from $y_t$ prior to estimation is algebraically equivalent to
projecting out the year fixed effects $\alpha_T$ in
equation~\eqref{eq:fe_seasonal}. The month coefficients recovered by
either approach are identical.

\begin{figure}[t!]
    \centering
    \begin{subfigure}{0.49\textwidth}
        \includegraphics[width=\linewidth]{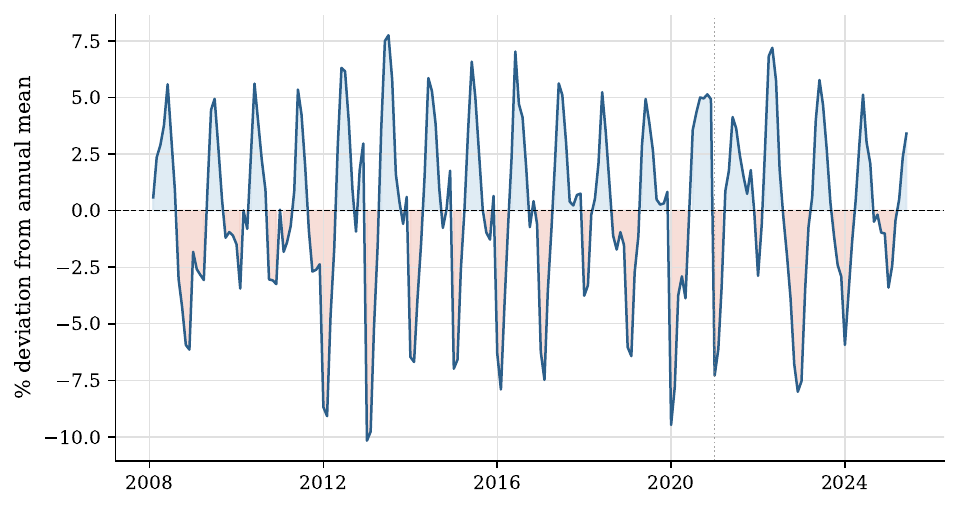}
        \caption{Real prices}
    \end{subfigure}\hfill
    \begin{subfigure}{0.49\textwidth}
        \includegraphics[width=\linewidth]{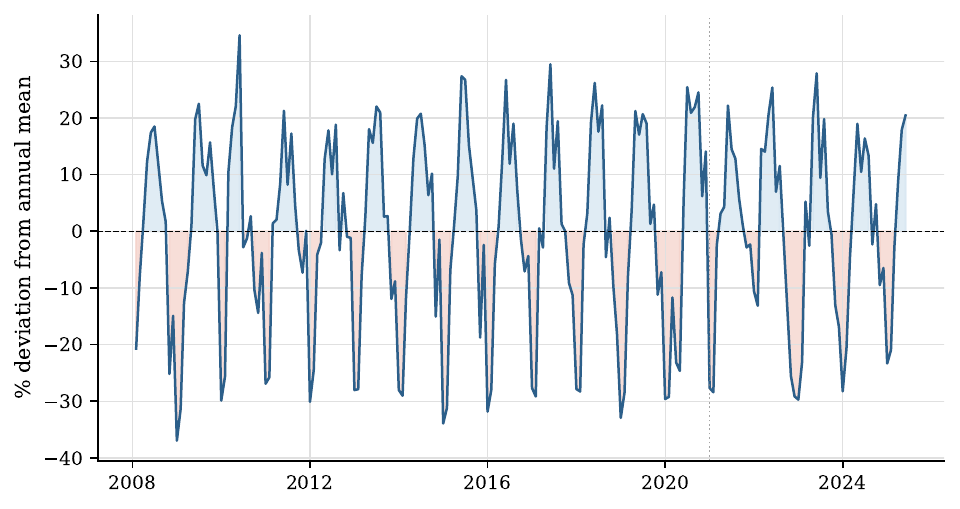}
        \caption{Sales}
    \end{subfigure}
    \caption{\small Seasonal components of US house prices and sales
    (Zillow). Each series plots $d_{m,T}$, the percentage deviation of
    month $m$ in year $T$ from that year's annual mean.}
    \label{fig:seasonalcomponents}
\end{figure}

Figure~\ref{fig:seasonalcomponents} plots $d_{m,T}$ for both series,
confirming strong and regular within-year cycles. Prices rise through
spring, reach a high point in early summer, and fall back in winter,
moving within a range of roughly $\pm6\%$ around the annual mean. Sales
volumes exhibit the same pattern but with a considerably larger amplitude,
oscillating between $-30\%$ and $+20\%$. Prices and sales move together,
rising through spring, peaking in summer, and falling through autumn and
winter.

The pattern is stable across most of the sample, but something changes
around 2021.\footnote{The choice of 2021 as the break date is not imposed.
In Appendix \ref{appendix:B} we scan all candidate break years from 2013 to 2023 using a
Chow $F$-test for stability of the seasonal components. The statistic is
small and insignificant at every candidate year from 2014 to 2020, then
spikes at 2021 for both prices and sales; see
Table~\ref{tab:break_tests} therein.} Figure~\ref{fig:Zillow-pre-post}
plots the average of $d_{m,T}$ separately for the pre-2021 and post-2021
periods and reveals a clear shift. In both series, the post-2021 profile
rises earlier in the year. Spring and late winter months strengthen, while summer and autumn weaken relative to pre-2021.

\begin{figure}[htbp]
    \centering
    \begin{subfigure}{0.48\textwidth}
        \includegraphics[width=\linewidth]{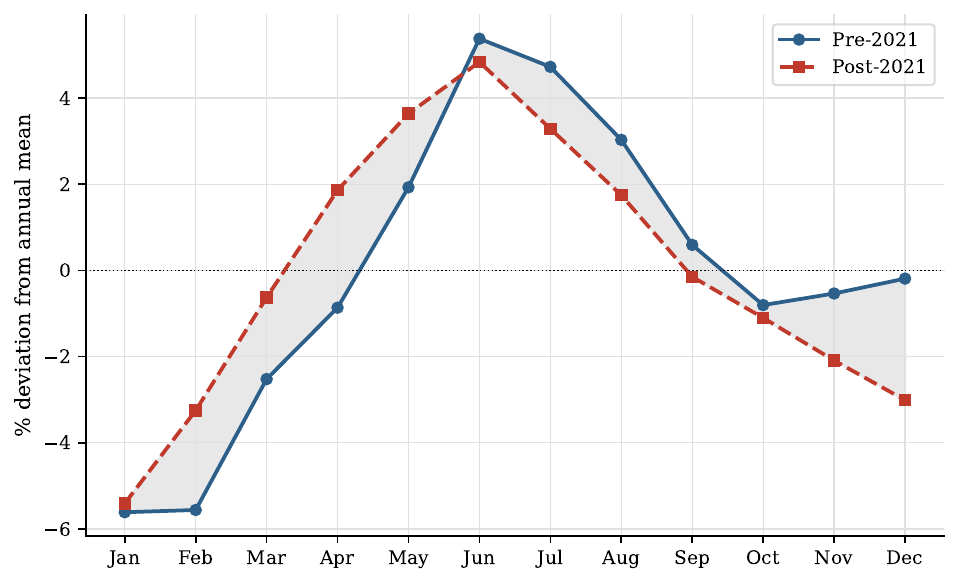}
        \caption{Prices}
    \end{subfigure}\hfill
    \begin{subfigure}{0.49\textwidth}
        \includegraphics[width=\linewidth]{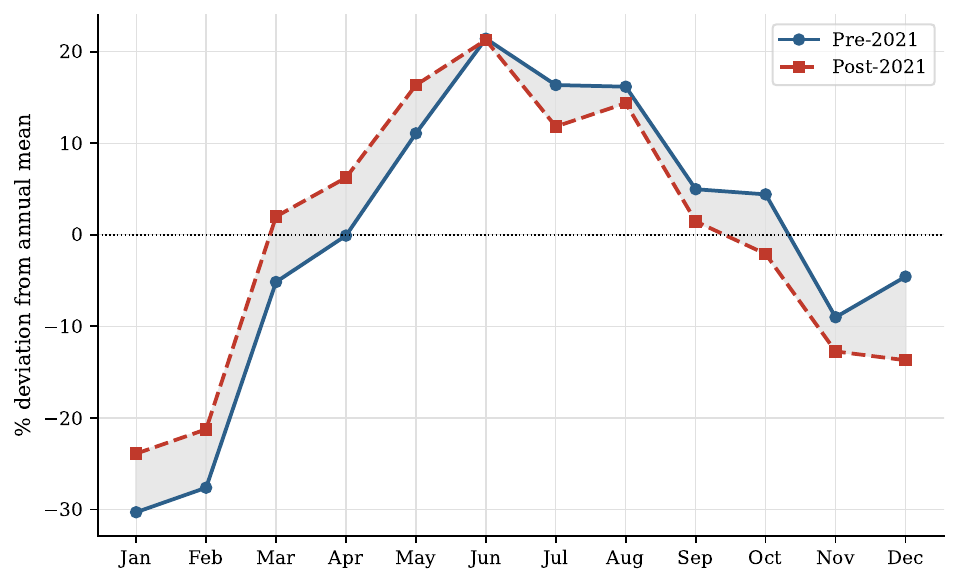}
        \caption{Sales}
    \end{subfigure}
    \caption{\small Seasonal components by month, pre- and post-2021
    (Zillow). Each series plots the average of $d_{m,T}$ within the
    respective period.}
    \label{fig:Zillow-pre-post}
\end{figure}

To assess whether this shift is statistically significant, we augment
equation~\eqref{eq:fe_seasonal} with a post-2021 indicator
$\text{POST}_t = \mathbf{1}\{t \geq \text{Jan 2021}\}$ and its
interactions with the month dummies,
\begin{equation}
y_t = \alpha_T + \gamma_m + \psi\,\text{POST}_t
    + \sum_{m=1}^{12}\mu_m \mathbf{1}\{m_t = m\}\cdot\text{POST}_t
    + \epsilon_t,
\label{eq:fe_shift}
\end{equation}
where $m_t$ denotes the calendar month of observation $t$, and we impose
$\sum_{m=1}^{12}\mu_m = 0$ for identification. We apply two complementary
tests. The results are reported in Table~\ref{tab:zillow_tests}.

The first test asks whether the seasonal profile changed after 2021. The evidence says yes.  A joint heteroskedasticity-robust $F$-test of
$H_0\colon \mu_m = 0$ for all $m$ \citep{stock_watson_2008}
rejects the null of a stable seasonal profile for both series
($p = 0.011$ for prices and $p = 0.007$ for sales).

The second test asks whether the change runs in the ``right direction": did the
first half of the year strengthen relative to the second? We test the
one-sided contrast $\bar\mu_{\text{H1}} - \bar\mu_{\text{H2}} > 0$,
where H1 denotes January--June and H2 denotes July--December, using
HC1-robust standard errors. The contrast is positive and significant for
both series ($p_1 = 0.013$ for prices and 
$p_1 = 0.010$ for sales), confirming that the shift runs from the second
half of the year toward the first.

The seasonal $\Delta$ columns in Table~\ref{tab:zillow_tests} show where
the reallocation is concentrated. For sales, the pattern is clean across
all four seasons: spring gains the most (+6.3\%), winter strengthens somewhat
(+1.2\%), while summer softens (--2.2\%) and autumn contracts the most
(--4.6\%). For prices the story is similar but winter is essentially flat,
with the gains concentrated in spring (+2.1\%) and the losses spread
across summer and autumn. Taken together, the evidence points to a
seasonal cycle that has shifted earlier in the calendar year: the active
season now begins sooner and ends sooner, with spring and late winter
strengthening at the expense of summer and autumn.\footnote{We have also
conducted a spring-share test comparing the contribution of spring months
to total seasonal variation before and after 2021, and individual
directional contrasts for spring gaining and summer losing. All are
statistically significant. We omit these from the table to avoid
overloading the paper with repetitive tests.}

\begin{table}[htbp]
\centering
\caption{Zillow: tests for a post-2021 seasonal shift}
\label{tab:zillow_tests}
\begin{tabular}{l cc cc cccc}
\toprule
 & \multicolumn{2}{c}{Joint $F$-test}
 & \multicolumn{2}{c}{Dir.\ contrast (H1--H2)}
 & \multicolumn{4}{c}{Seasonal $\Delta$ (pp)} \\
\cmidrule(lr){2-3}\cmidrule(lr){4-5}\cmidrule(lr){6-9}
Series & $F$ & $p$ & $t$ & $p_1$
       & Win.\ $\Delta$ & Spr.\ $\Delta$
       & Sum.\ $\Delta$ & Aut.\ $\Delta$ \\
\midrule
Prices & 2.31 & 0.011$^{*}$  & 2.25$^{*}$ & 0.013$^{*}$
       & $-0.1$ & $+2.1$ & $-1.1$ & $-0.9$ \\
Sales  & 2.45 & 0.007$^{**}$ & 2.34$^{*}$ & 0.010$^{*}$
       & $+1.2$ & $+6.3$ & $-2.2$ & $-4.6$ \\
\bottomrule
\end{tabular}%
\begin{tablenotes}
\small
\item The joint $F$-test examines $H_0\colon \mu_m = 0$ for all $m$,
where $\mu_m$ are the month$\times$post-2021 interaction coefficients from
equation~\eqref{eq:fe_shift} (HC1-robust). The H1--H2 directional
contrast tests the one-sided hypothesis $\bar\mu_{\textit{H1}} -
\bar\mu_{\textit{H2}} > 0$, where H1 denotes January--June and H2 denotes
July--December; $p_1$ is a one-sided $p$-value. Seasonal $\Delta$ columns
report post-minus-pre-2021 changes in average seasonal deviation (\%) for
each season: winter (December--February), spring (March--May), summer
(June--August), and autumn (September--November). Both series cover
February 2008 to June 2025 ($N=209$). $^{**}p<0.01$, $^{*}p<0.05$,
$^{\dagger}p<0.10$.
\end{tablenotes}
\end{table}

\medskip

\textit{Robustness.--- }As a robustness check, we re-estimate the seasonal components using
X-13 ARIMA-SEATS \citep{census_x13_2017}, the seasonal adjustment
programme developed by the US Census Bureau and used as the standard
tool by major statistical agencies worldwide. Unlike the annual-mean de-trending in equation~\eqref{eq:dmt}, X-13 separates the observed
series into trend, seasonal, and irregular components using an
iterative moving-average filter, with seasonal factors that are
allowed to evolve gradually over time. The resulting seasonal components, plotted in Figure~\ref{fig:zillowX13}, are strikingly similar to Figure~\ref{fig:seasonalcomponents}. Prices
oscillate within roughly $\pm 6\%$ around the trend and sales
between approximately $-30\%$ and $+20\%$, almost identical to what we have in the main text. 

Applying the same joint $F$-test and H1--H2 directional contrast to
the X-13 seasonal components confirms the same directional shift with stronger statistical significance (Table~\ref{tab:x13_tests}). The
seasonal $\Delta$ columns also show the same pattern as in 
Table~\ref{tab:zillow_tests}: spring strengthens, summer and autumn
weaken. Taken together, the X-13 results confirm that the shift toward earlier
seasonal activity is not an artefact of the annual-mean detrending
method, but a robust feature of the data. (Full results are reported in Appendix~\ref{appendix:B}.)

\section{Household Mobility Patterns}
\label{sec:moves}
The preceding section documented a shift in the seasonal pattern of prices and sales.
In the NT framework, these patterns arise from seasonal variation in
household mobility. If the seasonal profile of prices has shifted, then, all else
equal, the model predicts a corresponding shift in the timing of household moves.
Did mobility also shift earlier in the year? Evidence from three independent
sources (SIPP, Google Trends, and MoveHQ) suggests that it did.

\subsection{SIPP}
\label{subsec:sipp}
 
Our first source is the Survey of Income and Program Participation (SIPP),
administered by the U.S. Census Bureau. SIPP is a nationally representative
longitudinal survey designed to provide comprehensive information about the
income, programme participation, and economic well-being of individuals and
households in the United States. Respondents include all household members aged 15 and older, with interviews conducted via personal visits and telephone. Beginning with the 2018 panel, SIPP adopted an overlapping panel design in
which new panels are initiated each year and run concurrently.\footnote{The 2023 data
collection, for example, encompasses the first wave of the 2023 panel, the
second wave of the 2022 panel, the third wave of the 2021 panel, and the fourth of the 2020 panel, all referred to as the 2023
SIPP.} Crucially for our purposes, residential characteristics and household
composition are collected at a monthly frequency for each month of the
reference period, rather than only at the time of interview. SIPP also tracks
movers over time: when original sample members relocate to a new address,
interviewers attempt to locate and re-interview them at their new residence in
subsequent waves. The residence history module asks respondents to report the
calendar month in which they moved to their current address, as well as prior
addresses within the reference year. This allows us to identify the set of
unique movers across the panel and, critically, the calendar month in which
each move occurred.

\begin{figure}[htbp]
    \centering
    \begin{subfigure}{0.49\textwidth}
        \includegraphics[width=\linewidth]{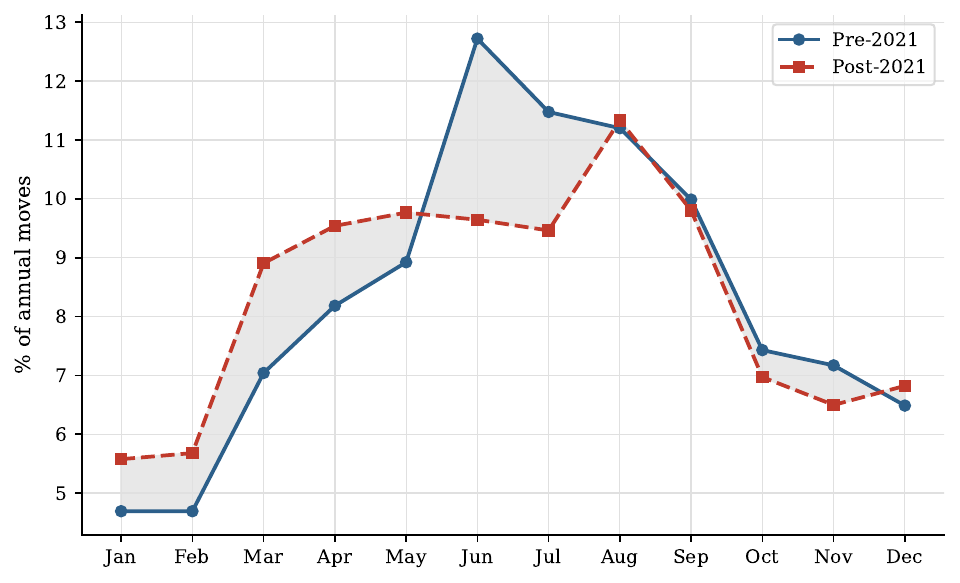}
        \caption{SIPP}
        \label{fig:Percentageofhosueholdmoves-panel-a}
    \end{subfigure}\hfill
    \begin{subfigure}{0.49\textwidth}
        \includegraphics[width=\linewidth]{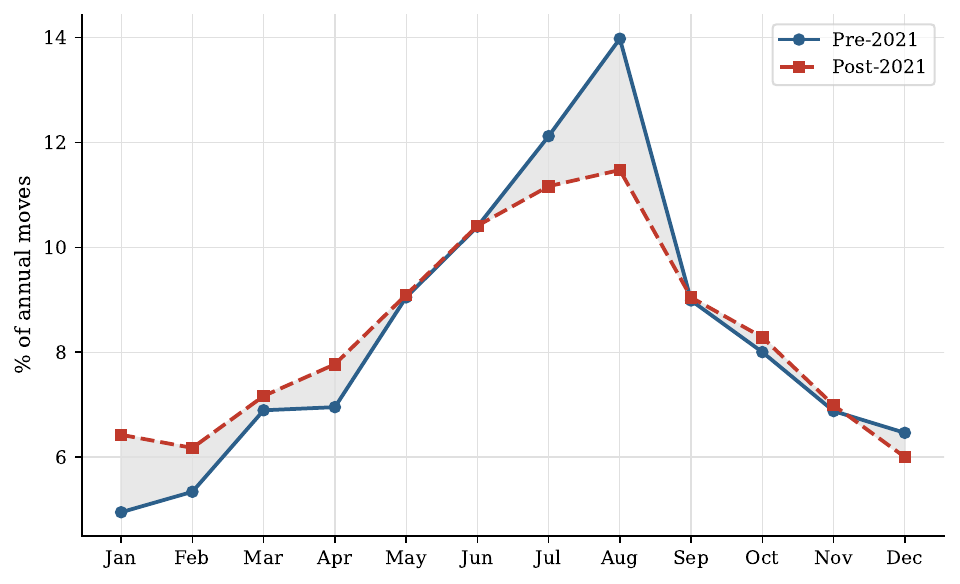}
        \caption{MoveHQ}
        \label{fig:Percentageofhosueholdmoves-panel-b}
    \end{subfigure}
    \caption{\small Percentage of household moves by month: pre-2021 and post-2021.}
    \label{fig:Percentageofhosueholdmoves}
\end{figure}

Figure~\ref{fig:Percentageofhosueholdmoves-panel-a} plots the average monthly
share of annual moves for the pre-2021 period (2017--2019) and the post-2021
period (2021--2023). The broad hump shape is present in both series. The key
difference, however, lies in the timing. Pre-2021, the distribution rises
steeply into June, which accounts for 12.7\% of all annual moves. Post-2021, June's share falls to 9.7\%, while March and April each gain around two
percentage points. The spring quarter as a whole absorbs 28.2\% of annual moves
post-2021, up from 24.1\% pre-2021, while the summer quarter contracts from
35.3\% to 30.5\%. The autumn and winter months remain largely unchanged. The moving pattern, just like house prices and sales, has shifted earlier in the calendar year.

\begin{table}[t!]
\centering
\caption{SIPP: Monthly distribution of household moves}
\label{tab:move_shares}
\begin{tabular}{lrrrr}
\toprule
 & {Pre-2021} & {Post-2021} & {Difference} & \\
Month & {(\%)} & {(\%)} & {(pp)} & {$t$-statistic} \\
\midrule
January &   4.7 &   5.5 &  0.79 &  0.25 \\
February &   4.7 &   5.6 &  0.92 &  1.78$^{*}$ \\
March &   7.1 &   8.9 &  1.80 &  2.80$^{***}$ \\
April &   8.1 &   9.5 &  1.38 &  1.93$^{*}$ \\
May &   8.9 &   9.8 &  0.86 &  1.31 \\
June &  12.7 &   9.7 & -3.02 & -3.74$^{***}$ \\
July &  11.4 &   9.5 & -1.91 & -2.35$^{**}$ \\
August &  11.3 &  11.4 &  0.12 &  0.17 \\
September &  10.0 &   9.9 & -0.12 & -0.16 \\
October &   7.4 &   7.0 & -0.46 & -0.83 \\
November &   7.1 &   6.5 & -0.63 & -1.06 \\
December &   6.5 &   6.8 &  0.28 &  0.45 \\
\midrule
Spring (Mar--May) &  24.1 &  28.2 &  4.03 &  3.25$^{***}$ \\
Summer (Jun--Aug) &  35.3 &  30.5 & -4.82 & -3.21$^{***}$ \\
\midrule
\multicolumn{5}{l}{Joint $F$-test: $F = 6.725$,\ $p < 0.001$} \\ 
\bottomrule
\end{tabular}
\begin{tablenotes}
\item \small Pre-2021 denotes the period 2017--2019; post-2021 denotes 2021--2023. Each column reports the weighted share of annual moves occurring in that month, averaged across years within each period. Shares are weighted using final person weights (\texttt{WPFINWGT}). $t$-statistics are from design-adjusted regressions using balanced repeated replication (BRR) with Fay correction ($\rho=0.5$), clustering at the sample-unit (\texttt{SSUID}) level. The joint $F$-statistic is the Rao--Scott design-adjusted test of homogeneity. $^{*}$$p<0.10$,\ $^{**}$$p<0.05$,\ $^{***}$$p<0.01$.
\end{tablenotes}
\end{table}

We formally test whether the seasonal distribution of moves changed after
2021 using survey-weighted Rao–Scott tests of equality of the monthly
distribution across periods. All estimates are weighted by the
final person weight (\texttt{WPFINWGT}), and standard errors are obtained
via balanced repeated replication using the 240 Census replicate weights
with a Fay correction factor of $\rho = 0.5$, as recommended for the SIPP's
multi-stage design \citep{SelectWeights2024}.
 
The results are reported in Table~\ref{tab:move_shares}. The joint $F$-test strongly rejects the null that the pre- and post-2021
monthly distributions are drawn from the same underlying seasonal profile
($F= 6.725$, $p < 0.001$). The month-by-month breakdown
confirms that the shift is concentrated in spring and late winter. February gains 0.92 percentage points, March gains 1.8, April gains 1.4, while June and July drop 3 and  1.9 percentage points. Taken together, the spring quarter absorbs 4 additional percentage points of annual moves post-2021, while the summer quarter contracts by about 4.8, all statistically significant. 

Before we move to the next section, note that we exclude the 2021 SIPP panel in our analysis. This panel covers reference year 2020, the acute phase of the pandemic, when lengthy lockdowns and major disruptions were still in place, and the Census Bureau issued a formal user note warning that the 2020 collection was compromised \citep{Census2020UserNote}.\footnote{In March 2020, field operations switched to telephone-only interviewing and remained so through the end of the collection period. Unit response rates fell sharply, leading the Census Bureau to conclude that the nonresponse weights for 2020 do not meet its statistical quality standards.} Although reference year 2020 is excluded from the main analysis, we include it as a robustness check in Appendix \ref{appendix:B}. Table~\ref{tab:move_shares_robustness} therein  reports results when 2020 is added to the pre-period, comparing 2017--2020 with 2021--2023. The results are very similar to the baseline specification: the spring gain and summer loss remain statistically significant, the joint F-test continues to reject homogeneity of the seasonal distribution, and the estimates for March, June, and July retain their sign and significance. February and April remain positive but are no longer individually significant.

\subsection{MoveHQ}

Our second source is a report by MoveHQ, a U.S. moving-industry software platform used by
professional moving and storage firms. The data in Figure \ref {fig:Percentageofhosueholdmoves-panel-b} come from the company's 2022
Moving Trends Report \citep{movehq_2022}, which draws on reportedly over one
million household moves recorded on its platform across 2019--2022, covering
both renters and homeowners. As the underlying microdata are proprietary, our analysis is based on the monthly figures reported in their publication.

Figure~\ref{fig:Percentageofhosueholdmoves-panel-b} plots the average monthly
distribution of moves for the pre-2021 period (2019--2020) and the post-2021
period (2021--2022). As in the SIPP data, the timing of moves shifts toward spring and away from
summer. August (the single busiest month pre-2021) declines noticeably in the post-2021 average, with July
showing a similar pattern. The earlier months, by contrast, strengthen: January through April register higher shares post-2021, and the report itself
highlights these as the months with the largest gains. This pattern closely mirrors the shift observed in the SIPP data and, coming from an independent source, corroborates the interpretation that the change reflects a broader shift in household mobility.

\begin{figure}[htbp]
    \centering
    \includegraphics[width=\linewidth]{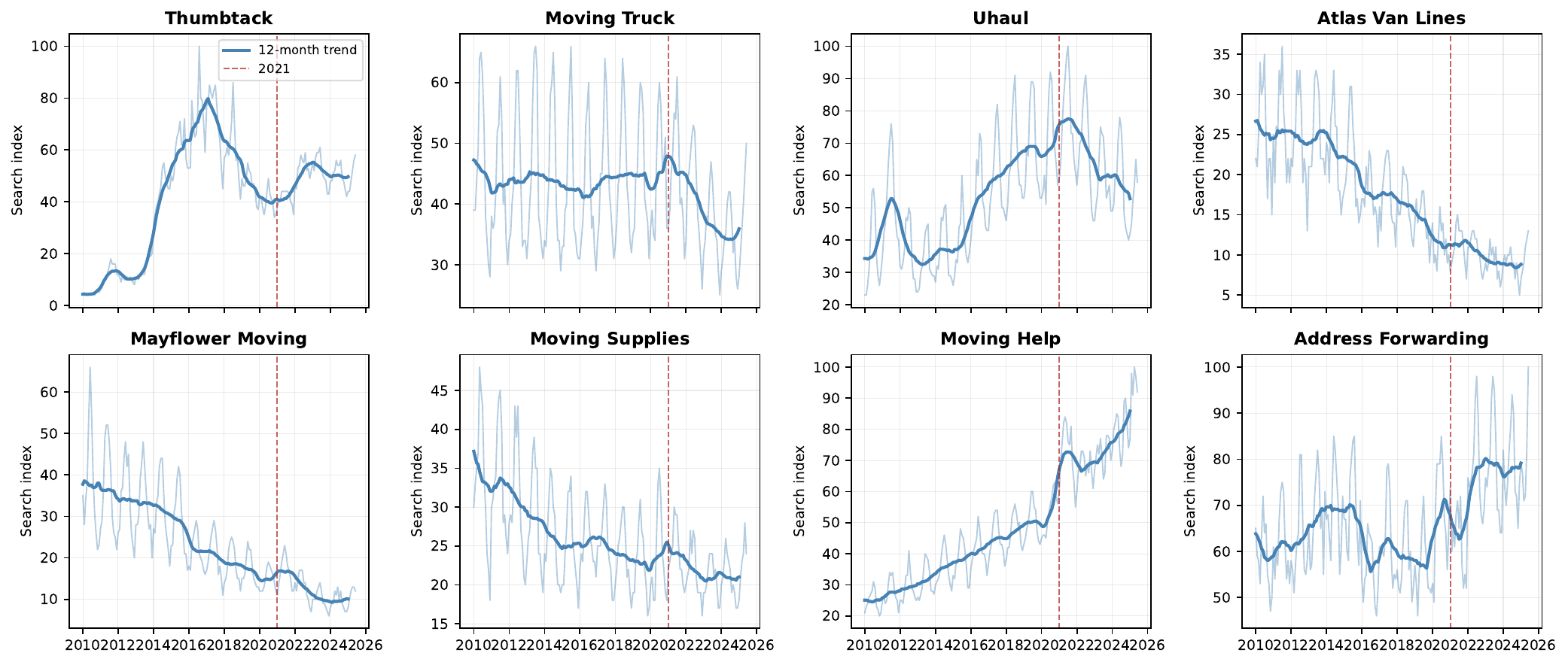}
    \caption{\small Raw monthly search interest and 12-month trend for
    eight moving-related keywords (Google Trends, US, 2010--2025). The dashed
    vertical line marks January 2021.}
    \label{fig:googletrends-raw}
\end{figure}

\subsection{Google Trends}\label{subsection:google}

Our third source is Google Trends search data. We obtain monthly search
interest indices from Google Trends for eight moving-related keywords:
\textit{moving truck}, \textit{uhaul}, \textit{atlas van lines},
\textit{mayflower moving}, \textit{moving supplies}, \textit{moving help},
\textit{address forwarding}, and \textit{thumbtack}. These terms cover many
aspects of a household move (vehicle rental, professional movers, packing
supplies, and administrative tasks) and collectively capture a broad
cross-section of moving-related search activity. The data cover January 2010
through June 2025 for the US. Figure~\ref{fig:googletrends-raw} plots raw
search interest alongside the 12-month trend for each keyword. All keywords
exhibit pronounced and highly regular seasonal spikes, with search activity
rising sharply during the summer months each year.

Let $g_{m,T}$ denote the Google Trends search interest index for calendar
month $m$ in year $T$, normalised by Google to the range $[0,100]$ within
each query window.\footnote{Because Google rescales each download relative to
the maximum observation in the requested window, indices from different
downloads are not directly comparable in levels. All analysis below is
therefore conducted in terms of percentage deviations from a within-series trend, which
are invariant to Google's rescaling.} For each keyword, we compute the percentage deviation of monthly search
interest from its 12-month rolling mean,
\begin{equation*}
    \tilde{g}_{m,T} \;=\; 100\,\frac{g_{m,T} - \bar{g}_{T}}{\bar{g}_{T}}.
\end{equation*}
This is the Google Trends analogue of the Zillow seasonal component $d_{m,T}$ defined in
equation~\eqref{eq:dmt}.

\begin{figure}[htbp]
    \centering
    \includegraphics[width=\linewidth]{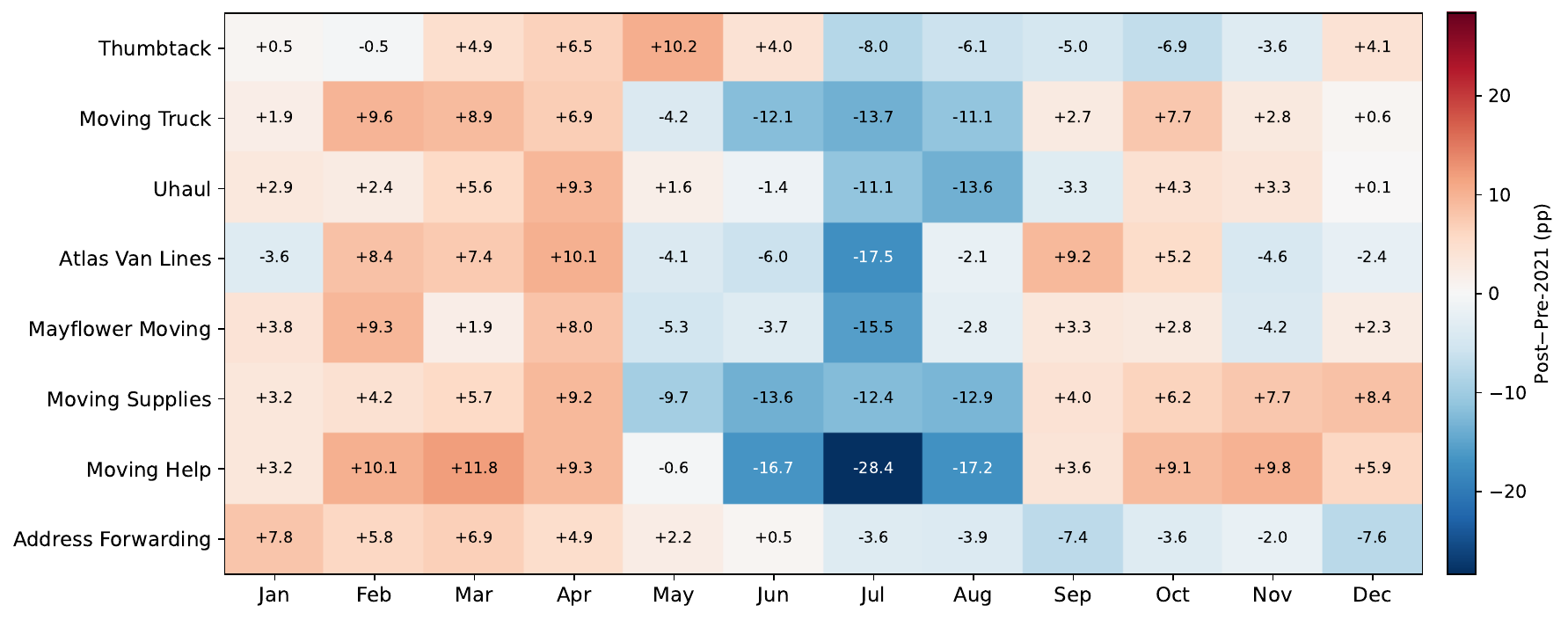}
    \caption{\small Each cell is the post-2021 minus pre-2021 difference in average search intensity for that keyword and month. Red indicates stronger seasonal search activity post-2021, blue weaker. Summer months are consistently blue across all keywords, indicating weaker search activity in the post-2021 period relative
to the pre-2021 period. Spring months, by contrast, show relatively stronger
activity post-2021.}
    \label{fig:googletrends-heatmap}
\end{figure}

We average $\tilde{g}_{m,T}$ by calendar month, separately for the pre-2021
period (2010--2020) and the post-2021 period (2021--2025), and compute the
difference. Figure~\ref{fig:googletrends-heatmap} plots these
post-minus-pre-2021 differences for all keywords and months simultaneously. Across all eight keywords, the spring months show positive differences, while
summer months show negative ones. The shift is largest for \textit{moving
help} and \textit{moving truck}, but the pattern holds for every keyword without exception. The autumn and winter months are mixed and smaller in magnitude.

To assess statistical significance, we estimate equation~\eqref{eq:fe_shift}
with $y_t = \tilde{g}_{m,T}$ as the dependent variable and
$\text{POST}_t = \mathbf{1}\{t \geq \text{Jan 2021}\}$ as the post-period indicator,
with year fixed effects $\alpha_T$ absorbing residual year-to-year variation
in search levels, and apply a heteroskedasticity-robust joint $F$-test of
$H_0\colon \mu_m = 0$ for all $m$. The results are reported in
Table~\ref{tab:google_ftest}. The null is rejected for all eight keywords at
the 5\% level. The pattern is consistent across all keywords and corroborates
the SIPP and MoveHQ findings through a third, independent data source.

\begin{table}[htbp]
\centering
\caption{Google Trends: heteroskedasticity-robust $F$-tests for seasonal
shift}
\label{tab:google_ftest}
\begin{tabular}{lccccc}
\toprule
Keyword & $N$ & $F$ & $p$-value & Spring $\Delta$ & Summer $\Delta$ \\
\midrule
Moving Truck       & 181 &  6.82 & $<$0.001$^{***}$ & $+$2.8 & $-$13.4 \\
Uhaul              & 181 &  3.35 & $<$0.001$^{***}$ & $+$5.0 & $-$9.2  \\
Atlas Van Lines    & 181 &  2.23 & 0.015$^{**}$    & $+$3.4 & $-$9.6  \\
Mayflower Moving   & 181 &  2.20 & 0.017$^{**}$    & $+$0.3 & $-$8.6  \\
Moving Supplies    & 181 &  7.20 & $<$0.001$^{***}$ & $+$1.5 & $-$13.2 \\
Moving Help        & 181 & 11.77 & $<$0.001$^{***}$ & $+$7.1 & $-$20.5 \\
Address Forwarding & 181 &  3.28 & $<$0.001$^{***}$ & $+$4.7 & $-$2.2  \\
Thumbtack          & 181 &  2.99 & 0.001$^{***}$    & $+$8.0 & $-$2.6  \\
\bottomrule
\end{tabular}
\begin{tablenotes}
\small
\item The dependent variable is $\tilde{g}_{m,T}$, the monthly percentage
deviation of search interest from the centred 12-month rolling mean. The $F$-statistic tests the joint
significance of eleven month$\times$post-2021 interaction terms from
equation~\eqref{eq:fe_shift}. Spring $\Delta$ and Summer $\Delta$ report
the change in average $\tilde{g}_{m,T}$ in percentage points.
$^{**}p < 0.05$, $^{***}p < 0.01$.
\end{tablenotes}
\end{table}

\paragraph{Proxying mobility rates.} Among the eight keywords, we select
\textit{moving truck} as the input for the model calibration in
Section~\ref{sec:calibration}. It is a generic term that does not mix brand loyalty with moving activity (unlike \textit{uhaul} or \textit{mayflower
moving}) and is tied directly to the physical act of moving house (unlike
\textit{address forwarding} or \textit{thumbtack}). To construct monthly move
share proxies, we compute each month's share of annual search volume for each
calendar year and average across years within each period:
\begin{equation*}
s_m^{j} = \frac{1}{|\mathcal{T}_j|}
    \sum_{T \in \mathcal{T}_j}
    \frac{g_{m,T}}{\sum_{m'=1}^{12} g_{m',T}},
    \qquad j \in \{\text{pre}, \text{post}\},
\end{equation*}
where $\mathcal{T}_{\text{pre}} = \{2010,\ldots,2020\}$ and
$\mathcal{T}_{\text{post}} = \{2021,\ldots,2025\}$. By construction,
$\sum_{m=1}^{12} s_m^{j} = 1$. These shares are fed into the model as the
seasonal distribution of moving activity.

\section{Model}
\label{sec:model}
What is the link between household mobility patterns and housing market seasonality?
To address this question, we recast the search-and-matching framework of
\citet{ngai_hot_2014} at a monthly frequency and calibrate it using the mobility
hazard rates implied by the moving patterns documented above. Calibrating the model at the monthly level allows us to move beyond the
two-period structure in NT and capture the finer seasonal dynamics observed in
the data. To our knowledge, this is the first attempt in the literature to
implement the NT framework at a monthly frequency.
\subsection{Environment}
Here we briefly outline the key elements of the search-and-matching framework of NT. A full description and detailed explanations can be found in the original paper. Consider a household at the beginning of month $m$, who holds a match of quality
$\varepsilon$. With probability
$\phi_{m+1}$ the match survives into the following month, while with
probability $1-\phi_{m+1}$ it dissolves and the household becomes a mover.
Let $H_m(\varepsilon)$ denote the value of holding a match of quality
$\varepsilon$ at the start of month $m$. This value satisfies the Bellman
equation
\begin{equation}
H_m(\varepsilon) = \varepsilon + \beta\big[\phi_{m+1} H_{m+1}(\varepsilon)
+ (1-\phi_{m+1}) X_{m+1}\big],
\label{Hm}
\end{equation}
where $\beta\in(0,1)$ is the monthly discount factor and $X_{m+1}$ is the
continuation value of entering month $m+1$ as a mover.

A mover is a household that simultaneously sells their current home and searches
for a new one. The value of entering month $m$ as a mover is denoted $X_m$.
Buyers in month $m$ sample match quality $\varepsilon$ from a distribution
$F_m(\varepsilon; v_m)$, where $v_m$ is the stock of houses listed for sale. A
thicker market raises the quality of available matches in the sense of
first-order stochastic dominance.

A mover who draws match quality $\varepsilon$ in month $m$ accepts the match
if and only if $\varepsilon \geq \varepsilon_m$. The reservation quality
$\varepsilon_m$ equates the value of accepting a match at the margin with the
value of remaining unmatched, collecting the flow payoff $u$,  and continuing the search next month:
\begin{equation}
H_m(\varepsilon_m) = \beta X_{m+1} + u.
\label{epsilon}
\end{equation}
It is convenient to write the net surplus from an accepted match as
\[
S_m(\varepsilon) = H_m(\varepsilon) - \beta X_{m+1} - u,
\]
which is zero at the cutoff and positive for accepted matches. Using this
notation, the value of being a mover can be expressed as
\begin{equation}
X_m = \beta X_{m+1} + u + \big[1-F_m(\varepsilon_m; v_m)\big]\,
\mathbb{E}[S_m(\varepsilon)\mid \varepsilon \geq \varepsilon_m].
\label{Xm}
\end{equation}
The first two terms capture the fallback value of waiting, while the final
term reflects the expected gains from trade today, weighted by the probability
of meeting a house above the cutoff.

The stock of vacancies evolves as listings roll over and new movers enter. A
fraction of last month's listings remain unsold, while households receiving a
moving shock in month $m$ contribute new vacancies:
\begin{equation}
v_m = 1-\phi_m + \phi_m\, v_{m-1}\, F_{m-1}(\varepsilon_{m-1}; v_{m-1}).
\label{vm}
\end{equation}
The number of transactions in month $m$ is
\begin{equation}
Q_m = v_m\,\big[1-F_m(\varepsilon_m; v_m)\big].
\label{Qm}
\end{equation}
Prices are determined through Nash bargaining. With bargaining weight $\theta$
on the seller, the average transaction price is
\begin{equation}
P_m = (1-\theta)\frac{u}{1-\beta} + \theta H_m(\varepsilon_m) +
\theta\,\mathbb{E}[S_m(\varepsilon)\mid \varepsilon \geq \varepsilon_m].
\label{Pm}
\end{equation}
The first component is constant across months; the seasonal variation comes
from the surplus term, which depends on vacancies through the match
distribution.

\begin{definition}
An equilibrium is a twelve-periodic object
$\{H_m, \varepsilon_m, X_m, v_m, Q_m, P_m\}_{m=1}^{12}$ such that
homeowner values $H_m$ satisfy \eqref{Hm}, reservation cutoffs $\varepsilon_m$
satisfy \eqref{epsilon}, mover values $X_m$ satisfy \eqref{Xm}, vacancies
$v_m$ satisfy \eqref{vm}, transactions $Q_m$ satisfy \eqref{Qm}, and prices
$P_m$ satisfy \eqref{Pm}.
\end{definition}

For analytic tractability, we assume that $\varepsilon$ is uniformly
distributed on $[0, v_m]$. The cdf associated with a higher $v_m$ first-order
stochastically dominates that with a lower $v_m$, which generates the
thin-thick market effect underlying the seasonal cycles.

\begin{proposition}
\label{exists and unique}
The housing market equilibrium exists and it is unique.
\end{proposition}
The proof is in Appendix \ref{appendix:A}. With existence and uniqueness established, we
now turn to the quantitative exercise of feeding the empirical hazard rates derived from monthly move rates derived earlier in Section \ref{sec:moves} into the model and asking what it predicts for prices and volume.

\subsection{Solution Algorithm}
We numerically compute the equilibrium by iterating on the damped mapping
$\mathcal{T}_\lambda$ used in the proof (see Appendix \ref{appendix:A} for the
definition of $\mathcal{T}_\lambda$):
\[
\mathcal{T}_\lambda(\mathbf{Z})=(1-\lambda)\mathbf{Z}
+\lambda\,\mathcal{T}(\mathbf{Z}), \qquad 0<\lambda\le1.
\]
Since the equilibrium exists and is unique, the convergence of this
iteration ensures that we obtain the correct solution. The damping factor
$\lambda$ affects only the convergence speed: smaller values slow down
updates but guarantee stability, while larger values accelerate
convergence at the risk of oscillations. We set $\lambda = 0.01$, which
satisfies the theoretical bound $\lambda < \bar\lambda$ for all
calibrations considered. Starting from an initial
$(\mathbf{X}^{(0)},\mathbf{v}^{(0)})$, for $t=0,1,2,\ldots$ we
iterate:
\begin{align*}
&D^{(t)}_m=\sum_r w_{m,r}\,X^{(t)}_{m+r},\\
&\varepsilon^{(t+1)}_m=\frac{\beta X^{(t)}_{m+1}+u-D^{(t)}_m}{A_m},\\
&v^{(t+1)}_m=(1-\lambda)v^{(t)}_m
+\lambda\big[1-\phi_m+\phi_m\,\varepsilon^{(t+1)}_{m-1}\big],\\
&\rho_m=\max\{0,\,v^{(t+1)}_m-\varepsilon^{(t+1)}_m\},\\
&X^{(t+1)}_m=(1-\lambda)X^{(t)}_m+\lambda\Big[\beta X^{(t)}_{m+1}
+u+\tfrac{A_m}{2}\,\dfrac{\rho_m^2}
{\max\{v^{(t+1)}_m,\underline v\}}\Big].
\end{align*}
We iterate until the sup-norm distance between successive iterates falls
below $10^{-5}$, at which point we treat the solution as
converged.\footnote{As a validation step, we replicate the bi-annual
specification in NT by setting $m=2$ and using their calibrated parameter
values. Our algorithm delivers the same steady-state objects reported in
their Appendix: the probability of sale is $0.31$ in summer and $0.25$ in
winter, and steady-state stocks are $v_w=0.167$ and $v_s=0.180$.}

\section{Calibration and Results}
\label{sec:calibration}

\subsection{Parametrisation}

Before presenting the results, we describe the parameter choices. The key parameter that determines the shape of the price and volume
distributions, as well as their amplitude, is the vector of monthly hazards $\{1-\phi_m\}_{m=1}^{12}$. A higher hazard rate means more households enter the market as buyers
and sellers, resulting in a thicker market, which in turn increases both
transaction volume and prices. We therefore back out the hazard rates from the
monthly moving rates obtained from the data.

Let $s_m$ denote the
share of all annual moves that take place in month $m$, normalised so that
$\sum_{m=1}^{12} s_m = 1$. These shares are obtained directly from the SIPP or MoveHQ
data, or proxied by the Google Trends keyword series. We recover the monthly  $1-\phi_m$ by imposing that hazards are
proportional to $s_m$. Specifically, we write
\begin{equation}
    1 - \phi_m = \kappa \, s_m, \qquad m = 1, \ldots, 12,
    \label{eq:hazard_proportional}
\end{equation}
where $\kappa > 0$ is a scalar common to all months. The proportionality
assumption ensures that the seasonal shape of the hazard vector is determined
entirely by the observed distribution of moves across months, while $\kappa$
controls the overall level. We identify $\kappa$ by requiring that the implied
annual moving probability is consistent with the observed annual move rate, $\eta$:
\begin{equation}
    1 - \eta \;=\;
    \prod_{m=1}^{12} \bigl(1 - \kappa \, s_m\bigr).
    \label{eq:annual_survival_monthly}
\end{equation}
Since the left-hand side is decreasing in $\kappa$ and equals one at
$\kappa = 0$, equation~\eqref{eq:annual_survival_monthly} has a unique
solution, which we pin down by iteration. The resulting 
$\{1-\phi_m\}_{m=1}^{12}$ is then fed directly into the model as the
calibrated hazard vector.

Estimates of annual mobility rates, $\eta$, are
available from CPS Migration tables
\citep{census_mobility}. The pre-2021 average  is $\eta = 10.3\%$, while post-2021 it is
$\eta = 8.3\%$.\footnote{The fall in $\eta$ is consistent with the decline in residential mobility documented in the literature \citep{MolloySmithWozniak2011,JiaEtAl2023}.}

The monthly discount factor $\beta=\hat \beta (1-\delta)$ reflects two considerations. The first is standard time discounting: Following NT, we use an annual interest rate of 6\%, which translates to a monthly discount factor $\hat\beta \approx 99.5\%$. The second is the risk that an ongoing transaction opportunity is disrupted before completion, due to failed negotiations, financing problems, inspection outcomes, or the collapse of housing chains. We capture this in reduced form with parameter $\delta$, which is the monthly probability of such a disruption. Empirical evidence suggests that these risks are economically meaningful: in the UK, roughly 25--30\% of residential property transactions fell through before completion in 2024, while in the US, depending on the state and city, the cancellation rates can reach around 20\% \citep{quickmovenow2024fallthrough,redfin2024cancellations}. Setting $\delta = 0.025$ implies an annual disruption probability of about 26\%, broadly in line with these figures. The quantitative role of $\delta$ is discussed further when we explore the amplitude of price calibration below.

Housing services $u$ are proxied by the rent-to-price ratio, which we set to 3\% to account for taxes. Rather than fixing a price level, we solve for
$u$ endogenously by requiring that $u = 0.03 \times \bar{P} / 12$, where
$\bar{P}$ is the average equilibrium price.\footnote{Since prices are endogenous, this is implemented through iteration. Starting from an initial guess, we solve the equilibrium,
update $u$ from the resulting average price, and repeat until convergence.} Finally, the seller's bargaining weight is set to $\theta =0.5$, giving equal bargaining power to buyers and sellers.

\subsection{Results}
We calibrate the model using two alternative proxies for the seasonal distribution of moves: the SIPP monthly move shares and the Google Trends \textit{moving truck} search index. In the first case, the move shares come directly from the SIPP estimates of monthly moves (plotted in Figure \ref{fig:Percentageofhosueholdmoves-panel-a}). In the second case, the implied move shares are calculated as outlined in Section~\ref{subsection:google}. These shares pin down the hazard vectors via equation~\eqref{eq:annual_survival_monthly}. We then feed these hazard rates into the model and solve for the equilibrium separately for each period. Figure~\ref{fig:calibration} reports the results.

\begin{figure}[t!]
\centering
    \begin{subfigure}{0.45\textwidth}
        \includegraphics[width=\linewidth]{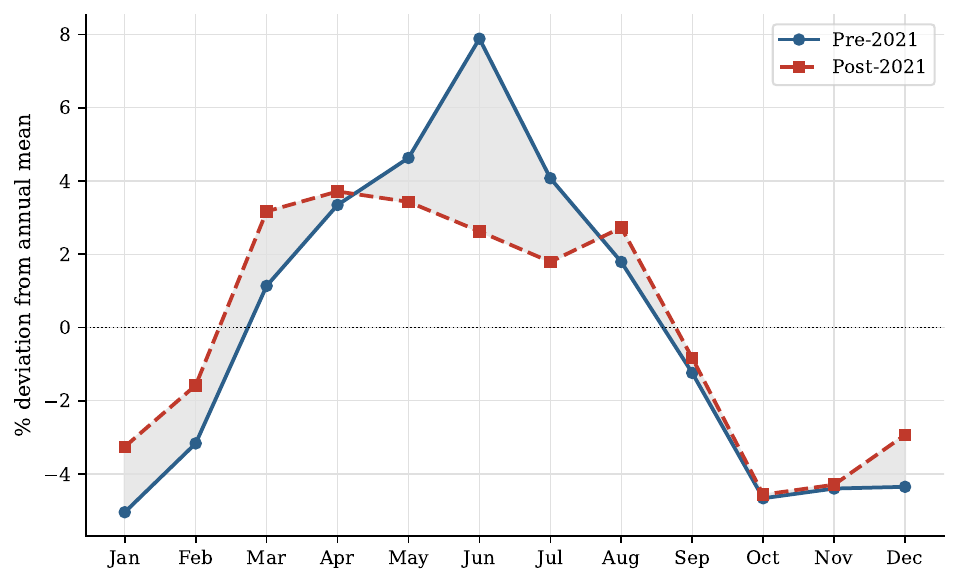}
        \caption{Prices (implied by SIPP)}
    \end{subfigure}\hfill
    \begin{subfigure}{0.45\textwidth}
        \includegraphics[width=\linewidth]{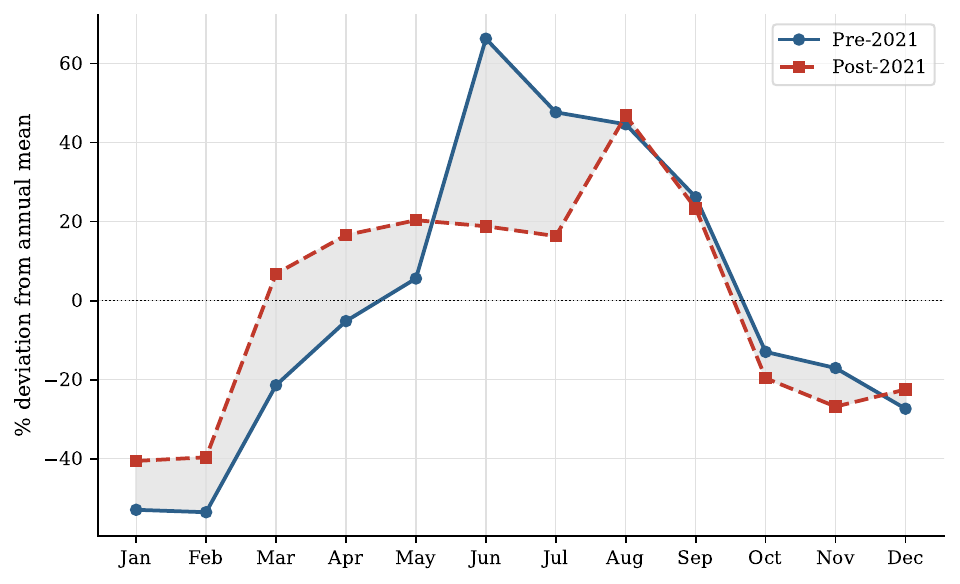}
        \caption{Sales (implied by SIPP)}
    \end{subfigure}

    \vspace{0.4cm}

    \begin{subfigure}{0.45\textwidth}
        \includegraphics[width=\linewidth]{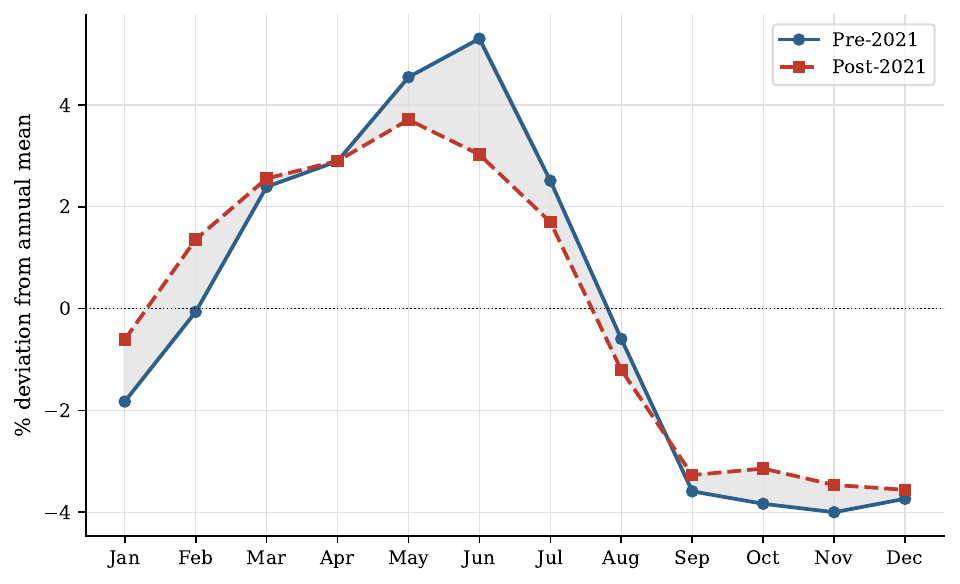}
        \caption{Prices (implied by Google Trends)}
    \end{subfigure}\hfill
    \begin{subfigure}{0.45\textwidth}
        \includegraphics[width=\linewidth]{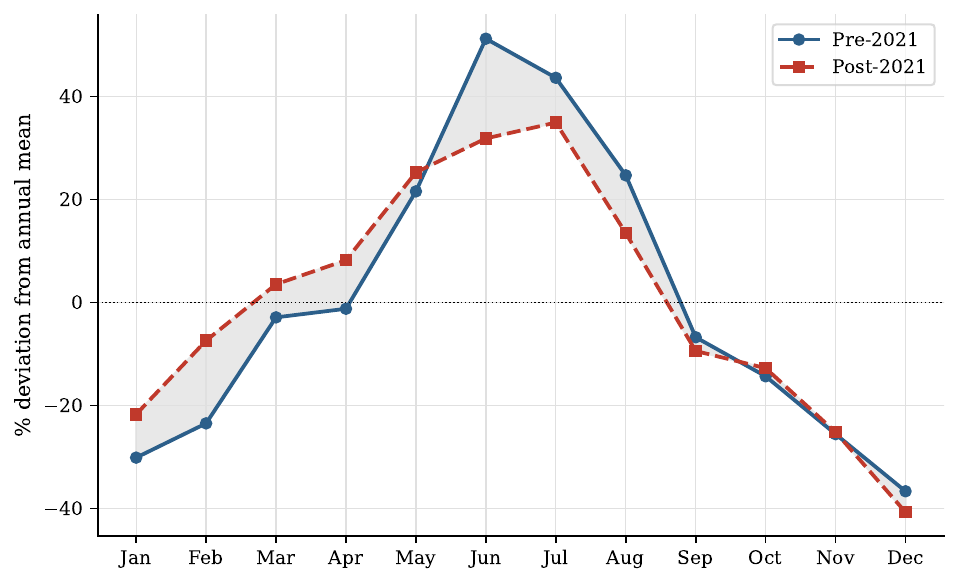}
        \caption{Sales (implied by Google Trends)}
    \end{subfigure}
    \caption{\small Model-implied seasonal deviations from the annual mean
    (\%), calibrated separately using SIPP monthly move shares (panels a--b)
    and the Google Trends \textit{moving truck} search index (panels c--d).
    Left panels show prices, right panels show transaction volume. Each
    series is computed as the percentage deviation of the monthly equilibrium
    value from its annual mean, averaged over the pre-2021 and post-2021 periods.}
    \label{fig:calibration}
\end{figure}

\medskip
\textit{The spring shift.---} The calibration delivers the central finding of the paper: a shift in the
seasonal distribution of household mobility translates into a corresponding
shift in the seasonal cycle of prices and transaction volumes. Both
calibrations reproduce the pre-2021 seasonal pattern, with prices and
volume peaking in summer and troughing in winter. Post-2021, both profiles
shift earlier: spring months strengthen and summer months soften in both
prices and volume, consistent with the Zillow evidence in
Figure~\ref{fig:Zillow-pre-post}.

When households move earlier in the year, the market thickens earlier,
more transactions close in spring, and the matches that form generate
larger surpluses and higher prices. A shift in mobility toward spring
therefore translates directly into a spring shift in both prices and
sales. The model requires no change in preferences, technology, or market
structure; the observed change in when households move is sufficient.

\medskip
\textit{Price amplitude.---} The magnitude of price seasonality in the calibration is broadly
consistent with the Zillow data (Figure \ref{fig:Zillow-pre-post}), where prices fluctuate by
roughly $\pm6\%$ around the annual mean. Our calibrations generate price oscillations of about
$\pm4$–$5\%$ (Figure \ref{fig:calibration}), placing the simulated amplitude in the same
ballpark as the data.

This match is not a free result. To see why, consider setting $\delta = 0$, so that
$\beta = \hat{\beta} \approx 99.5\%$ reflects only standard time discounting. The spring
shift remains (prices still move away from summer toward spring), but price deviations collapse
to roughly $\pm1\%$ around the annual mean, far too compressed to be consistent with the Zillow
data.\footnote{The spring 
shift is present at every value of $\delta$, including $\delta = 0$. 
The disruption adjustment is needed to match the \textit{level} of price  variation; it plays no role in generating the \textit{direction} of the 
seasonal shift.} The reason is that with such a high $\beta$, agents are very patient and place nearly
equal weight on current and future market conditions. This diminishes the variation in prices and
makes their seasonal trajectory nearly flat.

Matching the observed amplitude requires making agents discount the future more heavily,
and transaction disruption makes exactly this adjustment, albeit in reduced form. Deals that fall through
after months of negotiation impose real costs---legal fees, survey expenses, lost time---that
are not recovered when the household re-enters the market. In a fully structural model this
would enter the value function explicitly; here we capture it in reduced form through $\delta$,
which lowers $\beta$ below $\hat{\beta}$ as defined in Section~\ref{sec:calibration}.
Setting $\delta = 0.025$ is consistent with the empirical cancellation rates cited above,
and, conveniently, it produces price amplitudes in line with the Zillow data.
\medskip

\textit{Price timing.---}
Both calibrations reproduce the pre-2021 price peak in June. Post-2021,
however, the model-implied peak shifts earlier---to April under the SIPP
calibration and to May under the Google Trends calibration---while the
Zillow data continue to peak in June. This difference primarily reflects
the shape of the underlying mobility hazards rather than a failure of the
mechanism. In the model, prices are forward-looking: a mover's reservation
quality $\varepsilon_m$ depends on the continuation value $X_{m+1}$, which
incorporates expected market thickness in future months. When the
post-2021 hazard profile becomes relatively flat, as in the SIPP data, the
model anticipates sustained high market thickness and shifts the price
response toward the beginning of that plateau, producing an earlier peak.
By contrast, the Google Trends hazard retains a more pronounced hump after
2021, so the forward-looking adjustment is weaker and the peak occurs in
May, which is closer to the timing observed in the Zillow data.

\medskip 

\textit{Sales volume.---}
In the Zillow data, seasonal volume deviations range from roughly $-30\%$
to $+20\%$ around the annual mean, with spring months gaining and summer
months softening post-2021. Both calibrations reproduce this seasonal
pattern and the post-2021 shift. The SIPP calibration produces
fluctuations of $-55\%$ to $+66\%$, which is considerably larger than
observed in the data. The Google Trends calibration yields a narrower
range of $-41\%$ to $+35\%$, bringing the implied variation meaningfully
closer to the Zillow benchmark.

Further note that, unlike prices, volume tracks the hazard rate closely rather than leading
it. The listing stock is driven by the current flow of new entrants, which
tracks the hazard; the reservation cutoff varies little across months, so
volume timing is governed almost entirely by the hazard
vector.\footnote{Under the uniform distribution assumption, transaction
volume in month $m$ is $Q_m = v_m - \varepsilon_m$, where $v_m$ is the
vacancy stock and $\varepsilon_m$ is the reservation cutoff. The bargaining
parameter $\theta$ affects only how surplus is split between buyer and
seller, and drops out of $Q_m$ entirely. The discount factor $\beta$ and
housing services flow $u$ enter through $\varepsilon_m$, but the
reservation cutoff varies little across months, so the effect on volume amplitude and timing is negligible. The seasonal profile of volume
is therefore determined almost entirely by $v_m$, which is driven by the
current flow of new entrants and hence tracks the hazard vector very closely.}
Volume amplitude is therefore only as good as the input: the SIPP hazard
varies too much across months, producing swings roughly twice as large as
the data, while the smoother Google Trends hazard brings the amplitude
meaningfully closer to the Zillow benchmark.

Taken together, the evidence suggests that the Google Trends
\textit{moving truck} series provides the more informative calibration
input. For prices it produces a smoother hazard profile that avoids the
excessive front-loading generated by the flat SIPP plateau, and for sales
it delivers volume fluctuations closer to those observed in the Zillow
data. Although it is only a proxy for household moves, the Google Trends
measure yields the more consistent fit across both dimensions of the
calibration.

 \section{Conclusion}\label{sec:conclusion}

This paper starts out by documenting a change in the seasonal timing of activity in the
US housing market after 2021. For decades, prices and transaction volumes
followed a stable seasonal cycle, rising through the spring and peaking
in early summer. Using Zillow data, we show that this pattern shifted
after 2021, with both prices and sales rising earlier in the year and the
seasonal peak arriving sooner than in the pre-pandemic period.

To interpret this change, the paper draws on the search-and-matching
framework of NT, in which housing market seasonality
is driven by the timing of household mobility. In that framework, a shift
in the seasonal pattern of prices and transaction volumes should, all
else equal, reflect a similar shift in the timing of household moves. We
therefore ask whether mobility patterns have changed in a way consistent
with the earlier seasonal cycle observed in the housing market. Using
three independent sources (SIPP, Google Trends search activity related to moving,
and an industry report on relocation services) we find that the timing of
moves has indeed shifted earlier in the year. 

We then extend the NT framework from two seasons
to twelve calendar months and characterize the periodic equilibrium of
the housing market. When the empirical moving-hazard profiles are fed
into the model, the equilibrium seasonal cycle of prices and transaction
volumes shifts earlier in the year, consistent with the patterns observed
in the data. The results therefore suggest that the recent change in
housing market seasonality can be explained by a change in the timing of
household mobility, without requiring any change in preferences,
technology, or market structure.

What caused the shift in the timing of moves remains unclear. A plausible answer is the pandemic-era rise in remote work. By weakening the link
between residence and workplace, working from home may have allowed
households greater flexibility in choosing when to relocate, rather than
concentrating moves in the traditional summer window. This, however,
remains a conjecture. Most of the emerging literature on post-pandemic
mobility studies \textit{where} households relocated, while the question of \textit{why} their timing may have changed has received
little attention. By documenting the change in seasonal mobility and
embedding it in a tractable model of housing-market seasonality, this
paper aims to make the question more accessible to future research.

Whether the shift is permanent is equally uncertain. Remote work has begun
retreating from its post-2021 plateau as return-to-office mandates spread, and if that retreat continues, moves should drift back toward summer, and the traditional seasonal pattern should gradually reassert itself. Testing this prediction, as longer post-pandemic data become available, is a natural direction for future research.

\begingroup
\bibliographystyle{apalike}
\bibliography{References_Seasonality}
\endgroup

\strut

\noindent \textit{Declaration of AI-assisted technologies in the manuscript preparation process.---} During the preparation of this work the authors used ChatGPT (OpenAI) in order to assist with drafting and editing the manuscript prose. After using this tool, the authors reviewed and edited the content as needed and take full responsibility for the content of the published article.

\newpage

\appendix
\numberwithin{equation}{section}

\section{Existence and uniqueness of the equilibrium}
\label{appendix:A}
\noindent \textbf{Proof of Proposition \ref{exists and unique}.}

\noindent \textit{Preliminaries.---} Because the Bellman equation (\ref{Hm})
is linear in $H_{m+1}(\varepsilon)$ and depends on $\varepsilon$ only
through terms that are themselves linear, it follows that if
$H_{m+1}(\varepsilon)$ is affine in $\varepsilon$, then $H_m(\varepsilon)$
must also be affine. By backward induction over the twelve months,
$H_m(\varepsilon)$ is affine for all $m$. Define the slope and intercept
of $H_m$ as
\[
A_m \equiv \frac{\partial H_m(\varepsilon)}{\partial \varepsilon},
\qquad
D_m \equiv H_m(0).
\]
It follows that $H_m(\varepsilon) = A_m \varepsilon + D_m$. Matching
coefficients in (\ref{Hm}) gives two recursions:
\begin{equation}
A_m = 1 + \beta \phi_{m+1} A_{m+1} \quad \text{and} \quad
D_m = \beta \phi_{m+1} D_{m+1} + \beta (1-\phi_{m+1}) X_{m+1}.
\label{Am and Dm}
\end{equation}
Iterating on $A_m$ and imposing 12-periodicity yields
\begin{equation}
A_m = \frac{\displaystyle\sum_{s=0}^{11}\beta^{s}
\!\Big(\prod_{j=1}^{s}\phi_{m+j}\Big)}
{\displaystyle 1-\beta^{12}\Phi}, \quad \text{where} \quad
\Phi \equiv \prod_{j=1}^{12}\phi_j.
\end{equation}
Letting $\phi_{\max} \equiv \max_m \phi_m$, we have
\[
1 \le A_m \le \frac{1}{1-\beta\phi_{\max}}, \qquad
A_{\min} \equiv \min_m A_m \ge 1, \quad
A_{\max} \equiv \max_m A_m \le \frac{1}{1-\beta\phi_{\max}}.
\]
Similarly, iterating on $D_m$ yields
\[
(1-\beta^{12}\Phi)D_m
= \sum_{r=1}^{12}\beta^{r}\Big(\prod_{j=1}^{r-1}\phi_{m+j}\Big)
(1-\phi_{m+r})X_{m+r},
\]
so that
\begin{equation}
D_m = \sum_{r=1}^{12} w_{m,r}\,X_{m+r}, \qquad
w_{m,r} \equiv \frac{\beta^{r}\Big(\prod_{j=1}^{r-1}\phi_{m+j}\Big)
(1-\phi_{m+r})}{1-\beta^{12}\Phi}.
\label{Dm}
\end{equation}
Combining $H_m(\varepsilon) = A_m \varepsilon + D_m$ with (\ref{epsilon}),
we obtain the cutoff and surplus relationships
\begin{equation}
\varepsilon_m = \frac{\beta X_{m+1} + u - D_m}{A_m} \quad \text{and}
\quad S_m(\varepsilon) = A_m(\varepsilon - \varepsilon_m).
\label{em and Sm}
\end{equation}
Under the assumption that $\varepsilon \sim \mathrm{Unif}[0,v_m]$ in
month $m$, we have
\[
1-F_m(\varepsilon_m;v_m) = 1-\frac{\varepsilon_m}{v_m}, \qquad
\mathbb{E}[\varepsilon \mid \varepsilon \ge \varepsilon_m]
= \frac{\varepsilon_m + v_m}{2},
\]
and therefore
\[
\mathbb{E}[S_m(\varepsilon)\mid \varepsilon \ge \varepsilon_m]
= \frac{A_m}{2}\,(v_m-\varepsilon_m).
\]
Substituting these expressions into (\ref{Xm})--(\ref{Pm}) yields
\begin{align}
X_m &= \beta X_{m+1} + u
+ \frac{A_m}{2}\frac{(v_m-\varepsilon_m)^2}{v_m}, \label{Xm-unif}\\
v_m &= 1-\phi_m+\phi_m\,\varepsilon_{m-1}.\label{vm-unif}
\end{align}

\medskip

\noindent \textit{Defining the mapping $\mathcal{T}$.---} Pick coordinates
$(\mathbf{X},\mathbf{v})=\{(X_m,v_m)\}_{m=1}^{12}$ and choose bounds
\[
\underline v \equiv 1-\phi_{\max}>0, \qquad
\bar v \equiv \frac{1-\phi_{\min}}{1-\phi_{\max}}, \qquad
\underline X \equiv \frac{u}{1-\beta}, \qquad
\bar X \equiv \frac{u}{1-\beta}+\frac{A_{\max}}{2(1-\beta)}\,\bar v,
\]
and define the compact convex box $\mathcal{K} \equiv
[\underline X,\bar X]^{12}\times[\underline v,\bar v]^{12}$. Given any
$(\mathbf{X},\mathbf{v})\in\mathcal{K}$, define the mapping
$\mathcal{T}$ by the following steps:
\begin{enumerate}[label=(\roman*)]
\item Using equation (\ref{Dm}), compute $D_m$.
\item Via (\ref{em and Sm}), compute the provisional cutoff
$\varepsilon_m = (\beta X_{m+1}+u-D_m)/A_m$ and let
\[
\bar\varepsilon_m \equiv \min\{\max\{0,\varepsilon_m\},\,v_m\}.
\]
\item Update vacancies via (\ref{vm-unif}):
$\tilde v_m = 1-\phi_m+\phi_m\bar\varepsilon_{m-1}$.
\item Update movers' values via (\ref{Xm-unif}):
\[
\tilde X_m = \beta X_{m+1}+u+\frac{A_m}{2}
\frac{(v_m-\bar\varepsilon_m)^2}{v_m}.
\]
\end{enumerate}
By construction $0\le\bar\varepsilon_m\le v_m$, so
\[
\underline v \le \tilde v_m = 1-\phi_m+\phi_m\bar\varepsilon_{m-1}
\le \bar v, \qquad
\underline X \le \tilde X_m = \beta X_{m+1}+u+\frac{A_m}{2}
\frac{(v_m-\bar\varepsilon_m)^2}{v_m} \le \bar X.
\]
It follows that $(\tilde{\mathbf{X}},\tilde{\mathbf{v}})\in\mathcal{K}$.
Because all operations are algebraic and $v_m \ge \underline v > 0$ on
$\mathcal{K}$, the mapping $\mathcal{T}:(\mathbf{X},\mathbf{v})\mapsto
(\tilde{\mathbf{X}},\tilde{\mathbf{v}})$ is continuous on $\mathcal{K}$.

\medskip

\noindent \textit{Existence and Uniqueness.---} For any vector $\mathbf{Z}$,
define the sup norm $\|\mathbf{Z}\| \equiv \max_m |Z_m|$, and for pairs
let $\|(\mathbf{X},\mathbf{v})\| \equiv \max\{\|\mathbf{X}\|,
\|\mathbf{v}\|\}$. Pick a damping coefficient $\lambda \in (0,1]$ and
define the modified mapping
\[
\mathcal{T}_\lambda(\mathbf{Z}) \equiv (1-\lambda)\,\mathbf{Z} +
\lambda\,\mathcal{T}(\mathbf{Z}).
\]

\begin{lemma}
$\mathcal{T}$ and $\mathcal{T}_\lambda$ have identical fixed points.
\end{lemma}

\begin{proof}
If $\mathbf{Z} = \mathcal{T}(\mathbf{Z})$, then
$\mathcal{T}_\lambda(\mathbf{Z}) = (1-\lambda)\mathbf{Z} +
\lambda\,\mathcal{T}(\mathbf{Z}) = (1-\lambda)\mathbf{Z} +
\lambda\mathbf{Z} = \mathbf{Z}$. Conversely, if
$\mathcal{T}_\lambda(\mathbf{Z}) = \mathbf{Z}$, then
$(1-\lambda)\mathbf{Z} + \lambda\,\mathcal{T}(\mathbf{Z}) = \mathbf{Z}$,
which gives $\lambda\bigl(\mathcal{T}(\mathbf{Z})-\mathbf{Z}\bigr) = 0$.
Since $\lambda > 0$, it follows that $\mathcal{T}(\mathbf{Z}) =
\mathbf{Z}$.
\end{proof}

Now, for any two points $\mathbf{Z} = (\mathbf{X},\mathbf{v})$ and
$\mathbf{Z}' = (\mathbf{X}',\mathbf{v}')$ in $\mathcal{K}$, it is
straightforward to show that
\[
\big\|\mathcal{T}(\mathbf{Z})-\mathcal{T}(\mathbf{Z}')\big\|
\le \kappa\,\|\mathbf{Z}-\mathbf{Z}'\|,
\]
where
\[
\kappa \equiv \beta + \frac{A_{\max}}{A_{\min}}\big(\beta+W^*\big),
\qquad W^* \equiv \max_m \sum_{r=1}^{12} w_{m,r}.
\]
It follows that
\[
\big\|\mathcal{T}_\lambda(\mathbf{Z})-\mathcal{T}_\lambda(\mathbf{Z}')
\big\| \le \kappa_\lambda\,\|\mathbf{Z}-\mathbf{Z}'\|,
\]
where
\[
\kappa_\lambda \equiv \beta + \lambda\,\frac{A_{\max}}{A_{\min}}
\big(\beta+W^*\big).
\]
The term $\kappa_\lambda$ is decreasing in $\lambda$ and satisfies
$\lim_{\lambda\to 0}\kappa_\lambda = \beta < 1$. Hence, choosing
\begin{equation}
\lambda < \bar\lambda = \frac{1-\beta}{\tfrac{A_{\max}}{A_{\min}}
(\beta+W^*)} \label{lambdabar}
\end{equation}
ensures $\kappa_\lambda < 1$, so that $\mathcal{T}_\lambda$ is a
contraction on $\mathcal{K}$. By Banach's fixed-point theorem,
$\mathcal{T}_\lambda$ has a unique fixed point. By the above lemma, this is also
the unique fixed point of $\mathcal{T}$, which means that the housing market equilibrium exists and it is unique. $\blacksquare$

\pagebreak

\section{Robustness Checks}
\label{appendix:B}
\subsection{Prices and Sales: Structural Break Test}
The choice of 2021 as the break date is not imposed. To verify it,
we scan all candidate break years from 2013 to 2023 and apply a Chow
$F$-test at each, testing whether the full 12-month seasonal profile
differs across the pre- and post-break subsamples.
Table~\ref{tab:break_tests} reports the results.

\begin{table}[htbp]
\centering
\caption{Structural break test across candidate break years}
\label{tab:break_tests}
\begin{tabular}{l cc}
\toprule
Break year & Prices & Sales \\
\midrule
2013 & 1.65$^{\dagger}$ & 0.49 \\
2014 & 0.77 & 0.26 \\
2015 & 0.66 & 0.27 \\
2016 & 0.67 & 0.42 \\
2017 & 0.81 & 0.50 \\
2018 & 0.98 & 0.45 \\
2019 & 0.87 & 0.43 \\
2020 & 0.98 & 0.64 \\
\midrule
\textbf{2021} & \textbf{1.91$^{*}$} & \textbf{1.79$^{\dagger}$} \\
2022 & 2.96$^{**}$ & 2.66$^{**}$ \\
2023 & 0.38 & 0.51 \\
\bottomrule
\end{tabular}
\begin{tablenotes}
    \item \small Each row treats the indicated year as the candidate structural break. The Chow $F$-test compares a restricted model (common 12-month seasonal profile across the full sample) against an unrestricted model (separate profiles pre and post break). The dependent variable is the within-year percentage deviation from the annual mean. Both series cover February 2008 to June 2025 ($N=209$). $^{**}p<0.01$, $^{*}p<0.05$, $^{\dagger}p<0.10$.
\end{tablenotes}
\end{table}

The $F$-statistics are small and insignificant at every candidate year
from 2014 to 2020 for both series. At 2021 they rise sharply, rejecting
stability in the seasonal profile of prices at the 5\% level
($F = 1.91$, $p = 0.035$) and of sales at the 10\% level
($F = 1.79$, $p = 0.052$). The 2022 row is also significant, but 2021 is the earliest year at which both series reject.

\subsection{Prices and Sales: X-13 Seasonal Decomposition}

The seasonal components used in Section~\ref{sec: document shifts}
are computed as percentage deviations from the annual mean,
$d_{m,T} = 100(y_{m,T} - \bar y_T)/\bar y_T$.  As a robustness check, we re-estimate the
seasonal components using X-13 ARIMA-SEATS
\citep{census_x13_2017}, the seasonal adjustment programme
developed by the US Census Bureau and used as the standard tool by
major statistical agencies worldwide.

\begin{figure}[htbp]
    \centering
    \includegraphics[width=\linewidth]{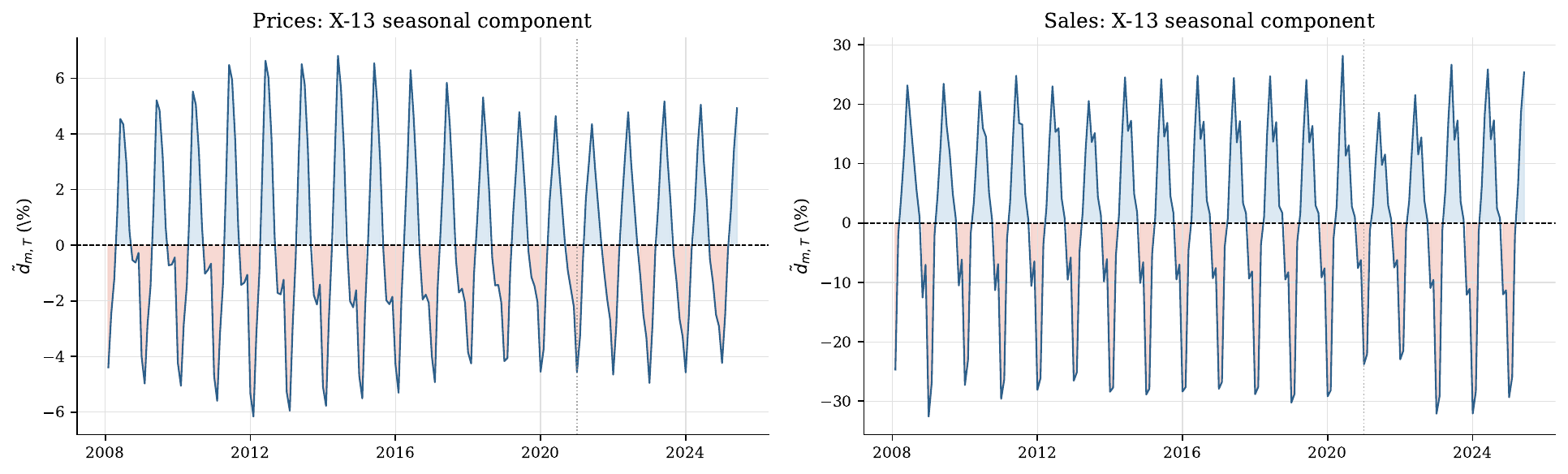}
    \caption{\small X-13 ARIMA-SEATS seasonal components for Zillow
    real prices (left) and sales (right). Each series plots
    $\tilde d_{t} = 100\,S_t/T_t$, the seasonal component as a
    percentage of the trend.}
    \label{fig:zillowX13}
\end{figure}

X-13 decomposes an observed series $y_t$ into trend, seasonal, and irregular components. We apply the X-13 procedure separately to the real price and sales series,
using the default specifications: additive components,  no trading-day adjustment and no
log transformation. Let $y_t^{sa}$ denote the seasonally adjusted series and
$T_t$ the trend component returned by the procedure. We then construct
\begin{equation*}
S_t = y_t - y_t^{sa}, \qquad \tilde d_t = 100\,\frac{S_t}{T_t},
\label{eq:x13dev}
\end{equation*}
so that the seasonal component is expressed as a percentage of the trend component.

Figure~\ref{fig:zillowX13} plots $\tilde d_t$ for prices and
sales. The regular seasonal spikes are clearly visible in both series and the two series are strikingly similar to their counterparts in the main text (Figure \ref{fig:seasonalcomponents}). Prices oscillate within a range of roughly $\pm 6\%$ around the
trend, and sales between approximately $-30\%$ and $+20\%$, almost
identical to the amplitudes in Figure \ref{fig:seasonalcomponents}.

We apply the same two tests as in Section~\ref{sec: document
shifts}  (the joint $F$-test and the H1--H2 directional contrast) to $\tilde d_t$ in place of $d_{m,T}$. The results are
reported in Table~\ref{tab:x13_tests}. The joint $F$-test rejects
the null of a stable seasonal profile for both series, with
considerably larger $F$-statistics than in
Table~\ref{tab:zillow_tests} ($F = 29.4$ for prices, $F = 6.4$
for sales). The H1--H2 directional contrast is also stronger
($t = 9.67$ for prices, $t = 3.22$ for sales), and the seasonal
$\Delta$ columns show the same pattern as before: spring
strengthens and summer and autumn weaken post-2021. The winter $\Delta$ is essentially zero for both series under this method,
consistent with the price result in Table~\ref{tab:zillow_tests}.

\begin{table}[htbp]
\centering
\caption{Robustness: X-13 ARIMA-SEATS seasonal decomposition}
\label{tab:x13_tests}
\small

\begin{tabular}{l cc cc cccc}
\toprule
 & \multicolumn{2}{c}{Joint $F$-test}
 & \multicolumn{2}{c}{Dir.\ contrast (H1--H2)}
 & \multicolumn{4}{c}{Seasonal $\Delta$ (\%)} \\
\cmidrule(lr){2-3}\cmidrule(lr){4-5}\cmidrule(lr){6-9}
Series & $F$ & $p$ & $t$ & $p_1$
       & Win.\ $\Delta$ & Spr.\ $\Delta$
       & Sum.\ $\Delta$ & Aut.\ $\Delta$ \\
\midrule
Prices & 29.37 & 0.000$^{**}$ & 9.67$^{**}$ & 0.000$^{**}$
       & $+0.2$ & $+1.5$ & $-1.3$ & $-0.3$ \\
Sales  & 6.39  & 0.000$^{**}$ & 3.22$^{**}$ & 0.002$^{**}$
       & $-0.2$ & $+2.1$ & $-1.0$ & $-0.8$ \\
\bottomrule
\end{tabular}%

\begin{tablenotes}
\small
\item The seasonal component $\tilde d_t = 100\,S_t/T_t$ is
estimated via X-13 ARIMA-SEATS, where $S_t$ is the
additive seasonal component and $T_t$ is the trend. The
joint $F$-test, H1--H2 directional contrast, and seasonal $\Delta$
columns are defined as in Table~\ref{tab:zillow_tests}. Both series
cover February 2008 to June 2025. $^{**}p<0.01$, $^{*}p<0.05$,
$^{\dagger}p<0.10$.
\end{tablenotes}
\end{table}

\subsection{Robustness Check for Table \ref{tab:move_shares}}
In the main text, the reference year 2020 is excluded from the pre-2021 period
because of the documented disruption to SIPP data collection during the acute
phase of the pandemic. This configuration corresponds to column~(A) in
Table~\ref{tab:move_shares_robustness}. To assess whether this exclusion
matters for the results, column~(C) adds 2020 to the pre-2021 period as a
robustness check.

The results are essentially unchanged. The spring gain and summer loss remain
statistically significant, the joint F-test continues to reject homogeneity of
the seasonal distribution, and March, June, and July retain both their sign
and statistical significance. February and April, while still positive, lose
individual significance, which appears consistent with 2020 introducing noise
into the pre-2021 period.

Column~(B) takes the opposite approach, dropping both COVID-adjacent years 2020 and 2021 from the analysis, and comparing 2017--2019 directly against 2022--2023. The pattern is identical to the baseline, and somewhat sharper: February and April recover significance, the spring difference rises to +4.5 percentage points, and the joint F-statistic increases to 7.28. The contrast strengthens when both transition years are removed from both sides of the comparison.

\begin{table}
\centering
\caption{Monthly move shares: robustness to sample period}
\label{tab:move_shares_robustness}
\begin{tabular}{lrrrrrr}
\toprule
 & \multicolumn{2}{c}{(A) Baseline} & \multicolumn{2}{c}{(B) 20 and 21 excluded} & \multicolumn{2}{c}{(C) All years} \\
\cmidrule(lr){2-3}\cmidrule(lr){4-5}\cmidrule(lr){6-7}
Month & {$\Delta$} & {$t$} & {$\Delta$} & {$t$} & {$\Delta$} & {$t$} \\
\midrule
January &  0.79 &  0.25 &  1.76 &  0.56 &  0.51 &  0.20 \\
February &  0.92$^{*}$ &  1.78 &  1.44$^{**}$ &  2.21 &  0.76 &  1.54 \\
March &  1.80$^{***}$ &  2.80 &  1.67$^{**}$ &  2.24 &  1.00$^{*}$ &  1.66 \\
April &  1.38$^{*}$ &  1.93 &  1.66$^{**}$ &  1.97 &  0.88 &  1.31 \\
May &  0.86 &  1.31 &  1.12 &  1.41 &  0.89 &  1.39 \\
June & -3.02$^{***}$ & -3.74 & -2.97$^{***}$ & -3.37 & -2.60$^{***}$ & -3.56 \\
July & -1.91$^{**}$ & -2.35 & -1.95$^{**}$ & -2.12 & -1.19$^{*}$ & -1.65 \\
August &  0.12 &  0.17 & -0.46 & -0.59 &  0.30 &  0.46 \\
September & -0.12 & -0.16 & -0.98 & -1.20 &  0.22 &  0.32 \\
October & -0.46 & -0.83 & -0.38 & -0.61 & -0.28 & -0.56 \\
November & -0.63 & -1.06 & -0.93 & -1.42 & -0.72 & -1.31 \\
December &  0.28 &  0.45 &  0.02 &  0.03 &  0.24 &  0.40 \\
\midrule
Spring (Mar--May) &  4.03$^{***}$ &  3.25 &  4.45$^{***}$ &  3.17 &  2.76$^{**}$ &  2.45 \\
Summer (Jun--Aug) & -4.82$^{***}$ & -3.21 & -5.37$^{***}$ & -3.30 & -3.49$^{***}$ & -2.76 \\
\midrule
Joint $F$ & \multicolumn{2}{c}{6.725, $p$=0.0000} & \multicolumn{2}{c}{7.275, $p$=0.0000} & \multicolumn{2}{c}{4.080, $p$=0.0000} \\
\bottomrule
\end{tabular}
\begin{tablenotes}
    \item \small (A) Baseline: Pre = 2017--19, Post = 2021--23; reference year 2020 excluded from the pre- period. (B) Pre = 2017--19, Post = 2022--23; years 2020 and 2021 excluded. (C) Pre = 2017--20, Post = 2021--23; all years inlcuded. $\Delta$ is change in the share of annual moves (pp, post minus pre). $t$-statistics are based on Balanced Repeated Replication (BRR) with Fay correction ($\rho=0.5$), using the 240 replicate weights provided by the Census Bureau. The joint $F$-statistic is the Rao--Scott design-adjusted test of whether the full monthly distribution shifted between periods. $^{*}$$p<0.10$,\ $^{**}$$p<0.05$,\ $^{***}$$p<0.01$.
\end{tablenotes}

\end{table}

\end{document}